\documentstyle[11pt]{article}
\input amssymb.sty
\title{Quantum computational logic with mixed states}

\author{\sc Hector
Freytes\thanks{Consejo Nacional de Investigaciones Cient\'ificas y
T\'ecnicas: Instituto Argentino de Matem\'atica (IAM), Saavedra 15 -
3er Piso - 1083  Buenos Aires, Argentina.} \thanks{Corresponding author Email: hfreytes@gmail.com}  $^{,1}$ and  Graciela
Domenech\thanks{Consejo Nacional de Investigaciones Cient\'ificas y
T\'ecnicas.} $^{,2}$ }

\date{{\small
1. Universit\`{a} di Cagliari, Dipartimento di Scienze Pedagogiche e Filosofiche, \\ Via Is Mirrionis 1,09123, Cagliari-Italia \\ 2. Instituto de Astronom\'ia  y F\'isica del Espacio (IAFE)\\
Casilla de Correo 67, Sucursal 28, 1428 Buenos Aires, Argentina}}

\begin{document}

\bibliographystyle{plain}

\maketitle

\begin{abstract}

\noindent Using an algebraic framework we solve a problem posed
in \cite{GIU1} and \cite{DGG} about the axiomatizability of a type
quantum computational logic related to fuzzy logic. A Hilbert-style
calculus is developed obtaining an algebraic strong completeness
theorem.

\end{abstract}

\begin{small}



\end{small}

\bibliography{pom}

\newtheorem{theo}{Theorem}[section]

\newtheorem{definition}[theo]{Definition}

\newtheorem{post}[theo]{Postulate}

\newtheorem{lem}[theo]{Lemma}

\newtheorem{prop}[theo]{Proposition}

\newtheorem{coro}[theo]{Corollary}

\newtheorem{exam}[theo]{Example}

\newtheorem{rema}[theo]{Remark}{\hspace*{4mm}}

\newtheorem{example}[theo]{Example}

\newcommand{\proof}{\noindent {\em Proof:\/}{\hspace*{4mm}}}

\newcommand{\qed}{\hfill$\Box$}

\section*{Introduction}

In the 1980s, Richard Feynman suggested that a quantum computer
based on quantum logic would evenly simulate quantum
mechanical systems. His ideas have spawned an active area of
research in physics which gave rise, in turn, to different logical
approaches to quantum computation. Quantum systems can simulate all computations which can be done by classical systems; however, one of the main advantages of quantum computation and quantum algorithms is that they can speed up computations \cite{KL}.

Standard quantum computing is
based on quantum systems with finite dimensional Hilbert spaces,
specially ${\mathbb{C}}^2$, the two-dimensional state space of a
{\it qbit}. A qbit state (the quantum counterpart of the classical
bit) is represented by a unit vector in ${\mathbb{C}}^2$ and,
generalizing for a positive integer $n$, {\it $n$-qbits} are pure
states represented by unit vectors in ${\mathbb{C}}^{2^n}$. They
conform the information units in quantum computation. These state
spaces only concerned with the ``static'' part of quantum computing
and possible logical systems can be founded in the Birkhoff and von
Neumann quantum logic based on the Hilbert lattices ${\cal
L}({\mathbb{C}}^{2^n})$  \cite{DUNN}. Similarly to the classical
computing case, we can introduce and study the behavior of a number
of {\it quantum logical gates} (hereafter quantum gates for short)
operating on qbits, giving rise to ``new forms" of quantum logic. These
gates are mathematically represented by {\it unitary operators} on
the appropriate Hilbert spaces of qbits. In other words, standard
quantum computation is mathematically founded on ``{\it
qbits-unitary operators}'' and only takes into account reversible
processes. This framework can be generalized to a powerful
mathematical representation of quantum computation  in which the
qbit states are replaced by {\it density operators} over Hilbert
spaces and unitary operators by linear operators acting over
endomorphisms of Hilbert spaces called  {\it quantum operations}.
The new model ``{\it density operators-quantum operations}'' also
called ``{\it quantum computation with mixed states}'' (\cite{AKN,
TA}) is equivalent in computational power to the standard one but
gives a place to irreversible processes as measurements in the
middle of the computation.

In \cite{GIU1} and \cite{DGG}, a quantum gate system called {\it
Poincar\'{e} irreversible quantum computational system} (for short
${\cal IP}$-system) was developed. Recently it was proved that the
mentioned quantum gates system can be seen in the framework given by
``density operators - quantum operations'' \cite{FLS}. The ${\cal
IP}$-system is an interesting quantum gates system specially for two
reasons: it is related to continuous $t$-norms and subsequent
generalizations allow to connect this system with sequential effect
algebras \cite{GU}, introduced to study the sequential action of
quantum effects which are unsharp versions of quantum events
\cite{GU1, GU2}.

Our study is motivated by the ${\cal IP}$-system, and mainly by the
following question proposed by the authors in \cite{GIU1} and
\cite{DGG}: \begin{center} {\it ``The axiomatizability of quantum
computational logic is an open problem.''} \end{center}

More precisely, in this paper we study the ${\cal IP}$-system from a
logic-algebraic perspective. A Hilbert-style calculus is develop
obtaining a strong completeness theorem respect to probabilistic
semantics associated with the ${\cal IP}$-system. The paper is
structured as follows: Section 1 contains generalities on universal
algebra and algebraic structures associated with fuzzy logic. In
Section 2, we briefly resume basic physical notions of  mathematical
approaches to quantum computation, with emphases in the approach of
``density operators - quantum operations''. This formalism allows to
build probabilistic models for quantum computational logics with
mixed states and provides the formal connection between the ${\cal
IP}$-system and  fuzzy logics based on continuous $t$-norms. In
Section 3, algebraic structures associated to quantum computation
are introduced. Specifically, we give an expansion of the equational
class known as square root quasi $MV$-algebras, expansion that we
call ``square root quasi $PMV$-algebra'' (or $\sqrt {qPMV}$-algebra
for short). In Section 4 we study a subvariety of $\sqrt
{qPMV}$-algebras called {\it Irreversible Poincar\'{e} Alegebras}.
They constitute the algebraic framework for the ${\cal IP}$-system.
Finally, in Section 5  we give a Hilbert-style calculus based on
probabilistic models related to the ${\cal IP}$-system and we
develop a ``non standard'' method of algebrization of this calculus
which allows to obtain an algebraic strong completeness theorem.

\section{Basic notions}

We freely use all basic notions of universal algebra that can be
found in {\rm \cite{BS}}. Let $\sigma$ be a type of algebras and let
${\cal A}$ be a class of algebras of type $\sigma$. For all algebras
$A, B$ in ${\cal A}$, $[A,B]_{\cal A}$ will denote the set of all
${\cal A}$-homomorphisms. An algebra $A$ in ${\cal A}$ is {\it
injective} iff for every monomorphism  $f \in [B, C]_{{\cal A}}$ and
every  $g \in [B,A]_{{\cal A}}$ there  exists  $h \in [C, A]_{{\cal
A}}$ such that $hf = g$. We denote by $Term_{{\cal A}}$ the
absolutely free algebra of type $\sigma$ built from the set of
variables $V = \{x_1,x_2, \ldots \}$. Each element of $Term_{{\cal
A}}$ is referred as an ${\cal A}$-term. For $t\in Term_{{\cal A}}$
we often write $t$ as $t(x_1,x_2, \ldots,x_n)$ to indicate that the
variables occurring in $t$ are among $x_1,x_2, \ldots,x_n$. Let $A
\in {\cal A}$. If $t(x_1,x_2, \ldots,x_n)\in Term_{{\cal A}}$ and
$a_1, \ldots a_n \in A$, by $t^A[a_1, \ldots, a_n ]$ we denote the
result of the application of the term operation $t^A$ to the
elements $a_1, \ldots a_n \in A$. A {\it valuation} in $A$ is a map
$v:V\rightarrow A$. Of course, any valuation $v$ in $A$ can be
uniquely extended to an ${\cal A}$-homomorphism $v:Term_{\cal A}
\rightarrow A$ in the usual   way, i.e., if $t_1, \ldots, t_n \in
Term_{\cal A}$ then $v(t(t_1, \ldots, t_n)) = t^A(v(t_1), \ldots,
v(t_n))$. Thus, valuations are identified with ${\cal
A}$-homomorphisms from the absolutely free algebra. If $t,s \in
Term_{\cal A}$, $ A \models t = s$ means that for each valuation $v$
in $A$, $v(t) = v(s)$ and ${\cal A} \models t = s$ means
that for each $A\in {\cal A}$, $A \models t = s$. \\

Now we introduce some basic notions in algebraic structures
associated to fuzzy logic. An {\it MV-algebra} {\rm \cite{CDM}} is an
algebra $ \langle A, \oplus, \neg,  0  \rangle$ of type $ \langle 2,
2, 0 \rangle$ satisfying the following equations:

\begin{enumerate}
\item[MV1]
$ \langle A, \oplus, 0  \rangle$ is an abelian monoid,

\item[MV2]
$\neg \neg x = x$,

\item[MV3]
$x \oplus \neg 0 = \neg 0$,

\item[MV4]
$\neg (\neg x \oplus y) \oplus y =  \neg (\neg y \oplus x) \oplus x$.

\end{enumerate}

\noindent
We denote by ${\cal MV}$ the variety of $MV$-algebras. In agreement with the usual $MV$-algebraic operations we define

\begin{enumerate}
\item[]
$x \odot y = \neg(\neg x \oplus \neg y)$, \hspace{1cm}  $x
\rightarrow y = \neg x \oplus y$,

\item[]
$x\land y = x \odot (x \rightarrow y)$,  \hspace{1cm} $x\lor y = (x
\rightarrow y) \rightarrow y$,

\item[]
$1 = \neg 0$.

\end{enumerate}

\noindent On each $MV$-algebra $A$ we can define an order $x \leq y $
iff $x \rightarrow y = 1$. This order turns $ \langle A, \land,
\lor, 0, 1  \rangle$ in a distributive bounded lattice with $1$ the
greatest element and $0$ the smallest element.

A very important example of $MV$-algebra is $[0,1]_{MV} = \{[0,1], \oplus, \neg, 0 \}$ such that $[0,1]$ is the real unit segment and $\oplus$ and $\neg$ are defined as follows: $$x\oplus y = \min(1, x+y) \hspace{1cm} \neg x = 1-x $$ The derivate operations in $[0,1]_{MV}$ are given by
$x\odot y = \max(0, x+y-1)$ (called \L ukasiewicz t-norm) and $x\rightarrow y = \min (1, 1-x+y)$. Finally the $MV$-lattice structure is the natural order in $[0,1]$.

Let $A$ be an $MV$-algebra. We define for all $a\in A$, $ \bigodot_1 a
= a$ and $\bigodot_{n+1} a = (\bigodot_n a) \odot a$. An element $a$
in $A$ is called {\it nilpotent} iff there exists a natural number
$n$ such that $\bigodot_n a = 0$.

\begin{prop}\label{SIMPMV} {\rm \cite[Theorem 3.5.1]{CDM}}
For every $MV$-algebra $A$, $A$ is simple iff $A$ is no trivial and
for each element $x<1$ in $A$, $x$ is a nilpotent element.

\qed
\end{prop}

A {\it product $MV$-algebra} {\rm \cite{MONT, MONT2, MR}} (for short: $PMV$-algebra) is
an algebra $ \langle A, \oplus, \bullet , \neg, 0  \rangle$ of type
$ \langle 2, 2, 1, 0 \rangle$ satisfying the following:

\begin{enumerate}
\item[1]
$ \langle A, \oplus, \neg,  0  \rangle$ is an $MV$-algebra,

\item[2]
$ \langle A, \bullet, 1 \rangle$ is an abelian monoid,

\item[3]
$x \bullet (y \odot \neg z) = (x \bullet y) \odot \neg (x \bullet z) $.
\end{enumerate}

We denote by ${\cal PMV}$ the variety of $PMV$-algebras. In each
$PMV$-algebra $A$ we also define for all $a\in A$, $a^1 = a$ and
$a^{n+1}= a^n\bullet a$. Important example of $PMV$-algebra is
$[0,1]_{MV}$ equipped with the usual multiplication. This algebra is
denoted by $[0,1]_{PMV}$. The following are almost immediate
consequences of the definition of $PMV$-algebras:

\begin{lem}
In each $PMV$-algebra we have

\begin{enumerate}
\item
$0\bullet x = 0$,

\item
If $a\leq b$ then $a\bullet x \leq b \bullet x$,

\item
$x\odot y \leq x\bullet y \leq x\land y$.
\end{enumerate}
\qed
\end{lem}

Two important subalgebras of $[0,1]_{PMV}$ are  ${\bf 2} = \{0,1\}$
and $G_{[0,1]}(\frac{1}{2})$ the sub $PMV$-algebra  generated by
$\frac{1}{2}$.

\begin{lem} \label{DENS}
$G_{[0,1]}(\frac{1}{2})$ is order dense in the real interval $[0,1]$.
\end{lem}

\begin{proof}
Let $a, b \in [0,1]$. We assume that $a < b < \frac{1}{2}$. Let $n_0
\in N$ the first natural such that $\frac{1}{2^{n_0}} \leq a$. Let
$n_1 \in N$ such that $\frac{1}{2^{n_1}}  \leq \frac{1}{4} \min \{b
- a,\hspace{0.2cm} a - \frac{1}{2^{n_0}}\}$. Thus there exists $n
\in N$ such that  $$a < s = \frac{1}{2^{n_0}} + \sum_n
\frac{1}{2^{n_1}}  = \frac{1}{2^{n_0}} + \bigoplus_n
\frac{1}{2^{n_1}} <b $$ and by construction $s \in G_{[0,1]}(\frac{1}{2})$. If
$\frac{1}{2} < a < b$ then $\neg b < \neg a < \frac{1}{2}$ since
$\neg \frac{1}{2} = 1 -\frac{1}{2}$. With the same argument we can
choice $s \in G_{[0,1]}(\frac{1}{2})$ such $\neg b < s < \neg a$. Thus $a < \neg s = 1-s
< b$. Hence $G_{[0,1]}(\frac{1}{2})$ is order dense in the real
interval $[0,1]$.

\qed
\end{proof}

\begin{prop} \label{CON}
\begin{enumerate}
\item
Each $PMV$-algebra is isomorphic to a subdirect product of linearly ordered $PMV$-algebras {\rm \cite[Lemma 2.3]{MONT2}}.

\item
Let $A$ be a $PMV$-algebra and let $B$ be the underlying $MV$-algebra.
Then $A$ and $B$ have the same congruences. Therefore $A$ is a
simple $PMV$-algebra iff $B$ is a simple $MV$-algebra {\rm \cite[Lemma
2.11]{MONT2}}.

\item
$[0,1]_{PMV}$ is injective in ${\cal PMV}$ {\rm\cite[\S3.2]{FRY}}.

\end{enumerate}
\qed
\end{prop}

\section{From physic to fuzzy logic}

\subsection{Quantum states}

The notion of {\it state of a physical system} is familiar from its
use in classical mechanics, where it is linked to the initial
conditions (the initial values of position and momenta) which
determine the solutions of the equation of motion of the system. For
any value of time, the state is represented by a point in the phase
space. In quantum mechanic the  description of the state  becomes
substantially modified. In fact, in quantum mechanics the state
embodies the specific history which preceded the instant to which
the state refers. As a simple description we may say that:

\begin{itemize}
\item[]
{\it A state is the result of a series of physical manipulations on the system which constitute the preparation of the state}
\end{itemize}

Quantum mechanics is founded in a set of simple postulates.
The first postulate gives a mathematical description of the concept of state and sets up the framework
in which quantum mechanics take places: the {\it Hilbert space}. In fact this postulate reads:  \\

\noindent {\bf Postulate}: A {\it closed physical system} is a
system which is totally isolated from the rest of the world.
Associated to any closed physical system is a complex Hilbert space
known as {\it the state space}. The {\it state} of a closed physical system
(or {\it pure state}) is wholly described by a unit vector in the state space.\\

In Dirac notation a pure state is denoted by $\vert \varphi \rangle$. A {\it quantum bit} or {\it
qbit}, the fundamental concept of quantum computation, is a pure
state in the Hilbert space ${\mathbb{C}}^2$. The standard
orthonormal basis $\{ \vert 0 \rangle , \vert 1 \rangle \}$ of
${\mathbb{C}}^2$ where $\vert 0 \rangle = (1,0)$ and $\vert 1 \rangle = (0,1)$ is called the {\it logical basis}. Thus, pure states $\vert \varphi \rangle$ in ${\mathbb{C}}^2$ are coherent superpositions of the
the basis vectors with complex coefficients
$$\vert \psi \rangle = c_0\vert 0 \rangle + c_1 \vert 1 \rangle,
\hspace{0.5cm} {\rm with} \hspace{0.5cm} \vert c_0 \vert^2 + \vert
c_1 \vert^2 = 1 $$

Quantum mechanics reads out the information content of a pure state
via the Born rule. By these means, a probability value is assigned
to a qbit as follows:

\begin{definition}{\rm \cite{DGG}, \cite{GIU1}}
{\rm Let $\vert \psi \rangle = c_0\vert 0 \rangle + c_1 \vert 1
\rangle$ be a qbit. Then its {\it probability value} is $p(\vert
\psi \rangle) = \vert c_1 \vert ^2$}
\end{definition}

The quantum states of interest in quantum computation lie in the
tensor product ${\otimes^n} {\mathbb{C}}^2 = {\mathbb{C}}^2 \otimes
{\mathbb{C}}^2 \otimes \ldots \otimes {\mathbb{C}}^2$ ($n$ times). The space
${\otimes^n}{\mathbb{C}}^2 $ is a $2^n$-dimensional complex space. A
special basis,   called the $2^n$-{\it computational basis}, is
chosen for ${\otimes^n} {\mathbb{C}}^2$. More precisely, it consists
of the $2^n$ orthogonal states $\vert \iota \rangle$, $0 \leq \iota
\leq 2^n$ where $\iota$ is in binary representation and $\vert \iota
\rangle$ can be seen as tensor product of states (Kronecker product) $\vert \iota
\rangle = \vert \iota_1 \rangle \otimes \vert \iota_2 \rangle
\otimes \ldots \otimes \vert \iota_n \rangle$ where $\iota_j \in
\{0,1\}$. A pure state $\vert \psi \rangle \in
{\otimes^n}{\mathbb{C}}^2$ is  a superposition of the basis
vectors $\vert \psi \rangle = \sum_{\iota = 1}^{2^n} c_{\iota}\vert
\iota \rangle$ with $\sum_{\iota = 1}^{2^n} \vert c_{\iota} \vert^2
= 1$.\\

In general, a quantum system is not in a pure state. This may be
caused, for example, by the non complete efficiency in the
preparation procedure and also by manipulations on the system as
measurements over pure states, both of which produce statistical
mixtures. Moreover, there are operations associated with partially
tracing out some degrees of freedom, which give rise to proper
mixtures. Besides, systems cannot be completely isolated from the
environment, undergoing decoherence of their states. Non pure
states, namely {\it mixed states}, are described by {\it density
operators}. A density operator is represented on the
$2^n$-dimensional complex Hilbert space by an Hermitian (i.e
$\rho^\dag = \rho$) positive operator with unit trace, $tr(\rho) = 1$.
In terms of density matrices, a pure state $\vert \psi \rangle$ can
be represented as a matrix product $\rho = \vert \psi \rangle
\langle \psi \vert$, where $\langle \psi \vert = \vert \psi
\rangle^{\dag}$. As a particular case, we may relate to each vector
of the logical basis of ${\mathbb{C}}^2$ one of the very important
density operators $P_0 = \vert 0 \rangle \langle 0 \vert$ and $P_1 =
\vert 1 \rangle \langle 1 \vert$ that represent the falsity-property
and the truth-property respectively. Due to the fact that the Pauli
matrices:

$$ \sigma_0 = I
\hspace{0.5cm} \sigma_x = \left(\begin{array}{cc}
0 & 1 \\
1 & 0
\end{array}\right)
\hspace{0.5cm} \sigma_y = \left(\begin{array}{cc}
0 & -i \\
i & 0
\end{array}\right)
\hspace{0.5cm} \sigma_z = \left(\begin{array}{cc}
1 & 0 \\
0 & -1
\end{array}\right)
$$

\noindent where $I = I^{(2)}$ is the $2\times 2$ identity matrix, are a
basis for the set of operators over ${\mathbb{C}}^2$, an arbitrary
density operator $\rho$ for $n$-qbits may be represented in terms of
tensor products of them in the following way:
$$\rho = \frac{1}{2^n} \sum_{\mu_1 \ldots \mu_n} P_{\mu_1 \ldots
\mu_n} (\sigma_{\mu_1}\otimes \ldots \otimes \sigma_{\mu_n})$$ where
$\mu_i \in \{0,x,y,x\}$ for each $i=1\ldots n$. The real expansion coefficients
$P_{\mu_1 \ldots \mu_n}$ are given by $P_{\mu_1 \ldots \mu_n} =
Tr(\sigma_{\mu_1}\otimes \ldots \otimes \sigma_{\mu_n}\rho)$. Since
the eigenvalues of the Pauli matrices are $\pm 1$, the expansion
coefficients satisfy $\vert P_{\mu_1 \ldots \mu_n} \vert \leq 1 $.

We denote by ${\cal D}({\otimes^n} {\mathbb{C}}^2)$ the set of all
density operators of ${\otimes^n}{\mathbb{C}}^2$, hence the set
${\cal D} = \bigcup_{i\in N} {\cal D}({\otimes^n} {\mathbb{C}}^2)$
will represent the set of all possible density operators. Moreover,
we can identify in each space ${\cal D}({\otimes^n}
{\mathbb{C}}^2)$, the two special operators $P_0^{(n)} =
\frac{1}{Tr(I^{n-1}\otimes P_0)} I^{n-1}\otimes P_0 $ and $P_1^{(n)}
= \frac{1}{Tr(I^{n-1}\otimes P_1)} I^{n-1}\otimes P_1 $ (where $n$
is even and $n \geq 2$) that represent in this framework, the
falsity-property and the truth-property respectively. By applying
the Born rule, the probability to obtain the truth-property
$P_1^{(n)}$ for a system being in the state $\rho$ is given by the
following definition:

\begin{definition}\label{DEFPROB} {\rm \cite{DGG}, \cite{GIU1}}
{\rm Let $\rho \in {\cal D}({\otimes^n} {\mathbb{C}}^2)$. Then its {\it
probability value} is $p(\rho) = Tr(P_1^{(n)} \rho)$.}
\end{definition}

\noindent Note that, in the particular case in which  $\rho = \vert
\psi \rangle \langle \psi \vert$ where $\vert \psi \rangle =
c_0\vert 0 \rangle + c_1 \vert 1 \rangle$,  we obtain that $p(\rho)
= \vert c_1 \vert ^2$. This definition of probability allows to
introduce  a binary relation $\leq_w$ on ${\cal D}$ in the following
way: $$\sigma \leq_w \rho \hspace{0.4cm} iff \hspace{0.4cm}
p(\sigma) \leq p(\rho)$$ One can easily see that $\langle {\cal D},
\leq_w \rangle$ is a preorder and it will play an important
role in the rest of the paper.

\subsection{Probabilistic models for quantum computational logics with mixed states}

In the usual representation of quantum computational processes, a
quantum circuit is identified with an appropriate composition of
{\it quantum gates}, i.e. unitary operators acting on pure states of
a convenient ($n$-fold tensor product) Hilbert space ${\otimes^n} {\mathbb{C}}^2$
\cite{NIC}. Consequently, quantum gates represent time reversible
evolutions of pure states of the system.

But for many reasons this restriction is unduly. On the one hand, it
does not encompass realistic physical states described by mixtures,
as mentioned above. On the other hand, there are interesting
processes that cannot be encoded in unitary evolutions, as
measurements in middle of the process. Several authors {\rm \cite{AKN}, \cite{DJFY}, \cite{GU}, \cite{TA}} have
paid attention to a more general model of quantum computational
processes, where pure states and unitary operators are replaced by
density operators and {\it quantum operations}, respectively. In
this case, time evolution is no longer necessarily reversible.

Let $H$ be a Hilbert space, ${\cal L}(H)$ be the vector space of all
linear operators on $H$ and ${\cal D}(H)$ be the set of density
operators. A {\it quantum operation} \cite{K} is a linear operator
${\cal E}:{\cal L}(H_1)\rightarrow {\cal L}(H_2)$ representable  as
${\cal E}(\rho)=\sum_{i}A_{i}\rho A_{i}^{\dagger }$ where $A_i$ are
operators satisfying $\sum_{i}A_{i}^{\dagger }A_{i}=I$ (Kraus
representation). It can be seen that a quantum operation maps
density operators into density operators. Every  unitary operator
${\cal U}$ on a Hilbert space $H$ gives rise to a quantum operation
${\cal O}_{\cal U}$ such that ${\cal O}_{\cal U}(\sigma) = {\cal
U}\sigma {\cal U}^{\dagger}$ for each $\sigma \in {\cal L}(H)$. Thus
quantum operations are a generalization of the model of quantum
computation based on unitary operators.

Quantum computational logics with mixed states may be presented as a
logic $\langle Term, \models\rangle $, where $Term$ is an absolute
free algebra, whose natural universe of interpretation is ${\cal D}$
and connectives are naturally interpreted as certain quantum
operations. More precisely, {\it canonical interpretations} are
$Term$-homomorphisms $e: Term \rightarrow {\cal D}$. To define a
relation of semantic consequence $\models$ based on the probability
assignment, it is necessary to introduce the notion of canonical
valuations. In fact, {\it canonical valuations} are functions over
the unitary real interval $f:Term \rightarrow [0,1]$ such that $f$
can be factorized in the following way:

\begin{center}
\unitlength=1mm
\begin{picture}(20,20)(0,0)
\put(8,16){\vector(3,0){5}}
\put(2,10){\vector(0,-2){5}}
\put(10,4){\vector(1,1){7}}

\put(2,10){\makebox(13,0){$\equiv$}}

\put(2,16){\makebox(0,0){$Term$}}
\put(20,16){\makebox(0,0){$[0,1]$}}
 \put(2,0){\makebox(0,0){${\cal D}$}} \put(2,20){\makebox(17,0){$f$}} \put(2,8){\makebox(-6,0){$e$}}
\put(18,2){\makebox(-4,3){$p$}}
\end{picture}
\end{center}

\noindent where $p$ is the probability function in the sense of
Definition \ref{DEFPROB}. We will refer to these diagrams as {\it
probabilistic models}. Then the semantical consequence $\models$
related to $\mathcal{D}$ is given by: $$\alpha \models \varphi
\hspace{0.3cm}  iff \hspace{0.3cm} {\cal R}[f(\alpha), f(\varphi)]$$
where ${\cal R} \subseteq [0,1]^2$ provides a relation between  $f(\alpha)$ and
$f(\varphi)$.

\subsection{Connection with fuzzy logic}

As a matter of fact, it can be shown {\rm \cite{DG}} that for some
systems of quantum operations (or quantum gates), this type of
semantics does not need to consider density operators other than
${\cal D}({\mathbb{C}}^2)$ for canonical models. This result smooths
things out to considerable extent for such systems, since density
operators in ${\mathbb{C}}^2$ are amenable to the well-known matrix
representation $$\rho = \frac{1}{2}(I + r_1\sigma_x +r_2\sigma_y +
r_3 \sigma_z)$$ where $r_{1},r_{2},r_{3}$ are real numbers such that
$ r_{1}^{2}+r_{2}^{2}+r_{3}^{2}\leq 1$. When a density operator
$\rho \in {\cal D}({\mathbb{C}}^2)$ represents a pure state, it can
be identified with a point $(r_1, r_2, r_3)$ on the sphere of radius
$1$ (Bloch sphere) and each $\rho \in {\cal D}({\mathbb{C}}^2)$ that
represents a mixed state with a point in the interior of the Bloch
sphere. We denote this identifications as $\rho = (r_1, r_2, r_3)$.
An interesting feature of density operators in ${\cal
D}({\mathbb{C}}^2)$ is the following: any real number $\lambda \in
[0,1]$, uniquely determines a density operator $\rho_\lambda$ given
by $$\rho_\lambda = (1-\lambda)P_0 + \lambda P_1$$

\begin{lem}\label{PROBBLOC1}{\rm \cite[Lemma 6.1]{DG}}
Let $\rho = (r_1,r_2,r_3) \in {\cal D}({\mathbb{C}}^2)$. Then we have:

\begin{enumerate}
\item
$p(\rho) = \frac{1-r_3}{2}$.

\item
If $\rho = \rho_\lambda$ for some $\lambda \in [0,1]$ then $\rho =
(0,0,1-2\lambda)$ and $p(\rho_\lambda) = \lambda$.

\end{enumerate}
\qed
\end{lem}

The connection  between quantum computational logic with mixed
states and fuzzy logic comes from the election of a  system of
quantum operations (or quantum gates) such that, when interpreted
under probabilistic models, they turn out in some kind of operation
in the real interval $[0,1]$ associated to fuzzy logic as continuous
t-norms \cite{HAJ}, left-continuous t-norms \cite{GE1}, etc.

The systems presented in {\rm \cite{GIU1} and \cite{DGG}}, precisely
those that motivate our study, are of this kind as will become clear
through the rest of the paper. It is not necessary  to consider
density operators other than ${\cal D}({\mathbb{C}}^2)$ for
canonical models (see {\rm \cite{DG}}). This quantum gates system
reduced to ${\cal D}({\mathbb{C}}^2)$ is the following:

\begin{itemize}
\item
$\sigma \oplus \tau = \rho_{p(\sigma) \oplus p(\tau)}$ \hspace
{4.2cm} [\L ukasiewicz gate]

\item
$\sigma \bullet \tau = \rho_{p(\sigma)\cdot p(\tau)}$ \hspace
{4.45cm} [IAND gate]

\item
$\neg \rho = \sigma_x \rho \sigma_x^\dag $ \hspace {5.2cm} [$NOT$
gate]

\item
$  \sqrt {\rho} = \left(\begin{array}{cc}
\frac{1+i}{2} & \frac{1-i}{2} \\
\frac{1-i}{2} & \frac{1+i}{2}
\end{array}\right)
 \rho \left(\begin{array}{cc}
\frac{1+i}{2} & \frac{1-i}{2} \\
\frac{1-i}{2} & \frac{1+i}{2}
\end{array}\right)^\dag
$ \hspace {1.05cm}[$\sqrt{NOT}$ gate]

\item
$P_1$, $P_0$, $\rho_{\frac{1}{2}}$ \hspace {4cm}\hspace {1.45cm}
[Constant gates]

\end{itemize}

\vspace{0.2cm}

We can see that quantum gates $ \bullet$, $ \sqrt{,}$ \hspace{0.1cm}
$\neg$ are quantum operations. The {\L}ukasiewicz quantum gate
$\oplus$ is not a quantum operation but it can be {\it
probabilistically approximated} in a uniform form by means of
quantum operations  \cite{FLS}. Thus we may introduce the following
algebraic system associated with the quantum gates known as {\it the
Poincar\'{e} irreversible quantum computational algebra} (for short
$IP$-algebra):
$$ \langle {\cal D}({\mathbb{C}}^2), \oplus, \bullet,  \neg, \sqrt, \hspace{0.1cm} P_0, \rho_\frac{1}{2}, P_1  \rangle $$

The following lemma provides the main properties of the $IP$-algebra that will be captured in an abstract algebraic framework.

\begin{lem}\label{PROBBLOC2}{\rm \cite[Lemma 6.1]{DG} and \cite[Lemma 3.7]{DF}}
Let $\tau, \sigma \in {\cal D}({\mathbb{C}}^2)$ and let $p$ be the
probability function over ${\cal D}({\mathbb{C}}^2)$. Then we have:

\begin{enumerate}
\item[1.]
$\langle {{\cal D}({\mathbb{C}}^2), \bullet} \rangle$ and $\langle {{\cal
D}({\mathbb{C}}^2), \oplus} \rangle$  are abelian monoids,

\item[2.]
$\tau \bullet P_0 = P_0$,

\item[3.]
$\tau \bullet P_1 = \rho_{p(\tau)}$,

\item[4.]
$p(\tau \bullet \sigma ) = p(\tau) p(\sigma)$,

\item[5.]
$p(\tau \oplus \sigma ) = p(\tau) \oplus p(\sigma)$,

\item[6.]
$\sqrt{\neg \tau} = \neg \sqrt{\tau}$,

\item[7.]
$\sqrt{\sqrt{\tau}} = \neg \tau$.

\end{enumerate}

Moreover if $\sigma = (r_1, r_2, r_3)$ then

\begin{enumerate}

\item[8.]
$\neg \sigma = (r_1, -r_2, -r_3) $ and $\sqrt{\sigma} = (r_1, -r_3,
r_2)$, hence $p(\neg \sigma) = \frac{1+r_3}{2}$ and
$p(\sqrt{\sigma}) = \frac{1-r_2}{2}$,

\item[9.]
$p(\sqrt{\tau \bullet \sigma}) = p(\sqrt{\tau \oplus \sigma}) = \frac{1}{2}$,

\item[10.]
$\frac {p(\sigma)}{4} \oplus \frac{p(\sqrt{ \sigma})}{4} \leq
\frac{1+\sqrt{2}}{4\sqrt{2}}$ \hspace{0.2cm} iff \hspace{0.2cm}
$r_2^2 + r_3^2 \leq 1$,

\item[11.]
$\frac {p(\sigma)}{4} \oplus \frac{1}{8} \leq \frac{3}{8} \leq
\frac{1+\sqrt{2}}{4\sqrt{2}}$.

\end{enumerate}\qed
\end{lem}

Recalling that in our case the assignment of probability is done via
a function $p:{\cal D}({\mathbb{C}}^2)\rightarrow [0,1]$, it is
possible to establish the following equivalence relation in  ${\cal
D}({\mathbb{C}}^2)$:
$$ \sigma \equiv \tau \hspace{0.5cm} iff \hspace{0.5cm} p(\sigma) =
p(\tau)$$ It is clear that this equivalence is strongly related to
the  preorder $\leq_w$ previously mentioned. Moreover it is not very hard
to see that $\equiv$ may be equivalently defined as
$$\sigma \equiv \tau \hspace{0.5cm} iff \hspace{0.5cm} \sigma \oplus P_0 = \tau \oplus P_0
\hspace{0.5cm} iff \hspace{0.5cm} \sigma \bullet P_1 = \tau \bullet
P_1$$ If we denote by $[\sigma]$ the equivalence class of $\sigma
\in {\cal D}({\mathbb{C}}^2)$, in view of Lemma \ref{PROBBLOC1} and
Lemma \ref{PROBBLOC2}, we can see that $$[\sigma] = [\sigma \bullet
P_1] = [\sigma \oplus P_0] = [\rho_{p(\sigma)}]$$ Thus, we can
consider the identification $({\cal D}({\mathbb{C}}^2)/_\equiv) =
(\rho_\lambda)_{\lambda \in [0,1]}$ and it may be easily proved that
$\langle ({\cal D}({\mathbb{C}}^2)/_\equiv), \oplus, \bullet, \neg,
[P_0], [P_1]  \rangle$ is a $PMV$-algebra,  ${\cal PMV}$-isomorphic
to $[0,1]_{PMV}$. The ${\cal PMV}$-isomorphism is given by the
assignment $[\rho_\lambda] \mapsto \lambda$. It is not very hard to
see that $\equiv$ is a $(\oplus, \bullet, \neg)$-congruence but not
a $\sqrt{,}$-congruence.

\begin{rema} \label{DIAGGEN}
{\rm It is important to remark that the notion of probability that
seems to be alien to a $IP$-algebra, is indeed represented by terms
of the algebra itself. More precisely, by  $x \oplus P_0$ or $x
\bullet P_1$. }
\end{rema}

Thus, any  algebraic abstract frame of the $IP$-algebra must be a class ${\cal A}$ of algebras
$ \langle A, \oplus, \bullet,  \neg, \sqrt, \hspace{0.1cm}
0, \frac{1}{2}, 1  \rangle $ of type $\langle2,2,1,1,0,0,0\rangle$,
such that it is able  not only to represent in an abstract form the
properties of  Lemma \ref{PROBBLOC2} but also is able to establish a
$(\oplus, \bullet, \neg)$-congruence $\equiv$ such that, $x\equiv y$ iff
$x\oplus 0 = y\oplus 0$ iff $x\bullet 1 = y\bullet 1$ satisfying
that $\langle A/_\equiv, \oplus, \bullet,  \neg, \sqrt, \hspace{0.1cm}
0, \frac{1}{2}, 1  \rangle$ is a $PMV$-algebra.

On the other hand, a logical calculus $\langle Term_{\cal A}, \models
\rangle $ interpreted in these algebraic generalization of the
$IP$-algebra will take into account the following  commutative
diagrams as a generalizations of the probabilistic models

\begin{center}
\unitlength=1mm
\begin{picture}(20,20)(0,0)
\put(8,16){\vector(3,0){5}} \put(2,10){\vector(0,-2){5}}
\put(10,4){\vector(1,1){7}}

\put(2,10){\makebox(13,0){$\equiv$}}

\put(1,16){\makebox(0,0){$Term_{\cal A}$}}
\put(20,16){\makebox(0,0){$A/_\equiv$}}
 \put(2,0){\makebox(0,0){$A$}} \put(2,20){\makebox(17,0){$f$}} \put(2,8){\makebox(-6,0){$e$}}
\put(18,2){\makebox(-4,3){$p$}}
\end{picture}
\end{center}

\noindent where $e$ is a $\langle \oplus, \bullet, \neg,
\sqrt,\hspace{0.1cm} 0, \frac{1}{2}, 1 \rangle$-homomorphism called
{\it interpretation on $A$}, $p$ is the natural $(\oplus, \bullet,
\neg, 0, \frac{1}{2}, 1 )$-homomorphism given by the $(\oplus, \bullet, \neg)$-congruence $\equiv$
(i.e an algebraic representation of the probability assignment) and the composition $f = pe$ is called {\it valuation}.
We will refer to these diagrams as {\it $PMV$-models} in ${\cal A}$. \\

In this paper we develop a logical system whose logical consequence
$\models$ is based on the preservation of the probability value
$p(\sigma) = 1$. More precisely, for each pair $\sigma, \tau \in
{\cal D}({\mathbb{C}}^2)$: $$\sigma \models \tau \hspace{0.3cm} iff
\hspace{0.3cm} p(\sigma) = 1 \Longrightarrow p(\tau) = 1 $$
Consequently, the generalization of the logical consequence
$\models$ in the $PMV$-models becomes: $\alpha \models \beta$ iff $f(\alpha)
= 1$ implies that $f(\beta)=1$ where $\alpha, \beta \in Term_{\cal A}$.

\begin{rema}
{\rm  The fact that the logical consequence of these systems is
related to functions $f$ factorized through the $PMV$-models, does
not allow to use standard methods of algebrization \cite{BlokPig1}
to study  the algebraic completeness of a Hilbert-style calculus.}
\end{rema}

\section{Quantum computational algebras}
The first and more basic algebraic structure associated to the
Poincar\'{e} system was introduced in {\rm \cite{LKPG1}} for the
reduced system $\langle \oplus, \neg, P_0, P_1  \rangle$. This is
the {\it quasi $MV$-algebra} or $qMV$-algebra for short. A
$qMV$-algebra is an algebra $ \langle A, \oplus, \neg, 0, 1
\rangle$ of type $ \langle 2, 1, 0, 0 \rangle$ satisfying the
following equation:

\begin{enumerate}
\item[Q1.]
$x \oplus (y \oplus z) = (x \oplus y) \oplus z $,

\item[Q2.]
$\neg \neg x = x$,

\item[Q3.]
$x \oplus 1 = 1$,

\item[Q4.]
$\neg (\neg x \oplus y) \oplus y = \neg (\neg y \oplus x) \oplus x$,

\item[Q5.]
$\neg(x \oplus 0) = \neg x \oplus 0$,

\item[Q6.]
$(x\oplus y) \oplus 0 = x\oplus y$,

\item[Q7.]
$\neg 0 = 1$.

\end{enumerate}

From an intuitive point of view, a $qMV$-algebra can be seen as
an $MV$-algebra which fails to satisfy the equation $x\oplus 0 = x$.
We denote by $q{\cal MV}$ the variety of $qMV$-algebras. We define the binary operations
$\odot, \lor, \land, \rightarrow$ in the same way as we did for $MV$-algebras.

\begin{lem}\label{PROPQMV}{\rm (\cite[Lemma 6]{LKPG1})}
The following equations are satisfied in each $qMV$-algebra:

\begin{enumerate}
\item[1.]
$x \oplus y = y \oplus x$,  \hspace{1cm} 5. \hspace{0.1cm} $x\oplus 0 = x \land x $,

\item[2.]
$x \oplus \neg x = 1$,  \hspace{1.4cm} 6. \hspace{0.1cm} $x\land y = y \land x $,

\item[3.]
$x \odot \neg x = 0$,  \hspace{1.4cm} 7. \hspace{0.1cm} $x\lor y = y \lor x $,

\item[4.]
$0\oplus 0 = 0$.

\end{enumerate}\qed
\end{lem}

In {\rm \cite{GLP}}, an abstract algebraic structure for the quantum gates
system $ \langle \oplus, \neg, \sqrt{,} \hspace{0.1cm}
 P_0, \rho_{\frac{1}{2}}, P_1  \rangle$ was introduced. These
algebras are known as {\it square root quasi $MV$-algebras} or {\it
$\sqrt{qMV}$-algebras} for short. A $\sqrt{qMV}$-algebra is an
algebra $ \langle A, \oplus, \neg, \sqrt{,} \hspace{0.1cm}
0, \frac{1}{2}, 1  \rangle$ of type $ \langle 2,1,1,0,0,0 \rangle$
such that:

\begin{enumerate}
\item[SQ1.]
$ \langle A, \oplus, \neg, 0, \frac{1}{2}, 1  \rangle$ is a $qMV$-algebra,

\item[SQ2.]
$\sqrt{\neg x} = \neg \sqrt x $,

\item[SQ3.]
$\sqrt {\sqrt x} = \neg x$,

\item[SQ4.]
$\sqrt{x\oplus y} \oplus 0 =  \sqrt{\frac{1}{2}} = \frac{1}{2}$.

\end{enumerate}

We denote by $\sqrt {q{\cal MV}}$ the variety of $\sqrt
{qMV}$-algebras. In what follows we will extend the structure of
$\sqrt {qMV}$-algebras considering an algebraic framework for the $IAND$
gate.

\begin{definition}
{\rm A  {\it $\sqrt {qPMV}$-algebra} is an algebra $\langle A,
\oplus, \bullet, \neg, \sqrt{,} \hspace{0.1cm}
 0,\frac{1}{2}, 1 \rangle$ of type $ \langle 2, 2, 1, 1, 0, 0, 0 \rangle$ satisfying the following:

\begin{enumerate}
\item
$ \langle A, \oplus,  \neg,  \sqrt{,} \hspace{0.1cm}  0, \frac{1}{2}, 1  \rangle$ is a $\sqrt {qMV}$-algebra,

\item
$x\bullet y = y \bullet x$,

\item
$x\bullet (y \bullet z) = (x\bullet y) \bullet z$,

\item
$ x \bullet 1 = x\oplus 0  $,

\item
$ x \bullet y = (x\bullet y) \oplus 0 $,

\item
$x \bullet (y \odot \neg z) = (x \bullet y) \odot \neg (x\bullet z)$,

\item
$\sqrt{x\bullet y} \oplus 0 = \frac{1}{2}$.

\end{enumerate}
}
\end{definition}

We denote by $\sqrt {q{\cal PMV}}$ the variety of $\sqrt {qPMV}$-algebras. It
is not very hard to see that the $IP$-algebra is a $\sqrt {qPMV}$-algebra. \\

Let $A$ be a $\sqrt {qPMV}$-algebra. Then we define a binary
relations $\leq$ on $A$: $$a\leq b \hspace{0.3cm} iff \hspace{0.3cm}
1 = a\rightarrow b$$ $$a\equiv b \hspace{0.3cm}  iff  \hspace{0.3cm}
a\leq b \hspace{0.2cm} and \hspace{0.2cm} b\leq a$$ It is clear that
$ \langle A, \leq  \rangle$ is a preorder. One can also easily prove
that $a \leq b$ iff $a\land b = a\oplus 0$ iff $a\lor b = b\oplus
0$. Moreover $a\equiv \hspace{0.1cm} (a\oplus 0)$.

\begin{prop}\label{CHIQPMV}
Let $A$ be a $\sqrt {qPMV}$-algebra and $a,b \in A$. Then we have:

\begin{enumerate}

\item
$a\bullet 0 = 0$,

\item
If $a\bullet b = 1$ then $a\oplus 0 = b\oplus 0 = 1$,

\item
If $a\leq b$ then $a \bullet x \leq b \bullet x $,

\item
$x\bullet y \leq x$,

\item
$x\bullet (y\oplus 0) = (x\bullet y) \oplus 0$,

\item
$\frac{1}{2} = \neg \frac{1}{2}$,

\item
$\frac{1}{2} \oplus 0 = \frac{1}{2}$,

\item
$\sqrt{x\oplus y} \oplus \sqrt{z\oplus w} = 1$.

\end{enumerate}
\end{prop}

\begin{proof}
1) $a \bullet 0 = a \bullet (0 \odot \neg 0 ) = (a \bullet 0) \odot
\neg (a \bullet 0) = 0$. \hspace{0.2cm} 2) Suppose that $a\bullet b
= 1$. Then $\neg(a\oplus 0) = 1\odot \neg(a \bullet 1)= (a\bullet b)
\odot \neg(a \bullet 1) = a\bullet (b\odot \neg 1)$ = 0. Thus $\neg
(a\oplus 0) = 0$, hence $a\oplus 0 = 1$. \hspace{0.2cm} 3) If $a\leq
b$ then $1 = a\rightarrow b = \neg(a\odot \neg b)$ and $0 = a\odot
\neg b$. Using   item 1. we have that $0 = x \bullet 0 = x \bullet
(a\odot \neg b) = (x \bullet a) \odot \neg (x \bullet b)$. Thus, $1
= \neg ((x \bullet a) \odot \neg (x \bullet b)) = (x \bullet a)
\rightarrow (x \bullet b) $ resulting $(x \bullet a) \leq (x \bullet
b)$. \hspace{0.2cm} 4) Since $x\leq 1$ by item 3. we have that
$x\bullet y \leq x\bullet 1 = x \oplus 0 \leq x$. \hspace{0.2cm} 5)
$x\bullet  (y\oplus 0) = x\bullet (y\bullet 1) = (x\bullet y)
\bullet 1 = (x\bullet y) \oplus 0$. \hspace{0.2cm} Items 6.,7. and
8. can be easily proved.

\qed
\end{proof}

\begin{definition}
{\rm Let $A$ be a $\sqrt {qPMV}$-algebra. An element $a\in A$ is
{\it regular} iff $a \oplus 0 = a$. We denote by $R(A)$ the set of
regular elements. }
\end{definition}

\begin{prop} \label{REG1}
Let $A$ be a $\sqrt {qPMV}$-algebra. Then we have:

\begin{enumerate}
\item
$\langle R(A), \oplus, \bullet, \neg, 0, \frac{1}{2}, 1  \rangle$ is a
$PMV$-algebra.

\item
$\equiv$ is a $\langle \oplus, \bullet, \neg \rangle$-congruence on
$A$ and $\langle A/_\equiv, \oplus, \bullet, \neg, [0], [\frac{1}{2}], [1] \rangle$ is a $PMV$-algebra.

\item
$A/_\equiv$ is ${\cal PMV}$-isomorphic to $R(A)$. This isomorphism is given by the assignment $[x]\mapsto x\oplus 0$.

\end{enumerate}
\end{prop}

\begin{proof}
1)From {\rm \cite[Lemma 9]{LKPG1}} $\langle R(A), \oplus, \neg, 0, 1
\rangle$ is an $MV$-algebra. Using Proposition \ref{CHIQPMV}-5, the
operation $\bullet$ is closed in $R(A)$. Now from the axioms of the
$\sqrt {qPMV}$-algebras, $\langle R(A), \oplus, \bullet, \neg,
 0, \frac{1}{2}, 1  \rangle$ results a $PMV$-algebra.

2) It is easy to see that $\equiv$ is a $\langle \oplus, \neg
\rangle$-congruence. For technical details see {\rm \cite{LKPG1}}.
From Proposition \ref{CHIQPMV}-3, $\equiv$ is compatible with
$\bullet$. For the second part it is clear that we only need to see
that the class $[1]$ is the identity in $\langle A/_\equiv , \bullet,
[1] \rangle$. In fact, $[x]\bullet [1] = [x\bullet 1] = [x\oplus 0]
= [x]$.

3) Since  $[x] = [x\oplus 0]$ for each $x\in A$, then $\varphi$ is
injective. If $x\in Reg(A)$ then $x = x\oplus 0$. Therefore
$\varphi([x]) = x \oplus 0 = x$ and $\varphi$ is surjective. Using
Proposition \ref{CHIQPMV}-5 we have that $\varphi ([x]\bullet [y]) =
\varphi([x\bullet y]) = (x\bullet y) \oplus 0 = (x\oplus 0) \bullet
(y\oplus 0) = \varphi([x]) \bullet \varphi([y])$. In the same way we
can prove that $\varphi ([x]\oplus [y]) = \varphi ([x]) \oplus
\varphi ([y])$. By axiom Q5 $\varphi(\neg [x]) = \neg \varphi([x])$
and $\varphi([c]) = c $ for $c = 0,1, \frac{1}{2}$ since they are
regular elements in $A$. Thus $[x]\mapsto x\oplus 0$ is a
${\cal PMV}$-isomorphism.
\qed
\end{proof}

\begin{rema}
{\rm From Proposition \ref{REG1} we can see that the natural
$\langle \oplus, \bullet, \neg \rangle$-homomorphism $A \rightarrow
A/_\equiv$ (equivalently represented as $A \rightarrow Reg(A)$ such
that $x\mapsto x\oplus 0$) is an abstract version of the notion of
probability in the $PMV$-model as  the remark \ref{DIAGGEN} and the
paragraph below it claim.}
\end{rema}

\begin{prop} \label{ST}
Let $ \langle A, \oplus, \bullet, \neg, 0, \frac{1}{2}, 1 \rangle$
be a $PMV$-algebra such that $\neg \frac{1}{2} = \frac{1}{2}$.
Consider the set $S_A = A \times A$ with the following operations:

\begin{enumerate}
\item[]
$(a,b) \oplus (c,d): = (a\oplus c, \frac{1}{2}) $,   \hspace{1.75 cm}
$0: = (0,\frac{1}{2})$

\item[]
$(a,b) \bullet (c,d): = (a\bullet c, \frac{1}{2} )$,  \hspace{2 cm}
$1: = (1,\frac{1}{2})$

\item[]
$\neg (a,b): = (\neg a, \neg b)$, \hspace{3 cm} $\frac{1}{2}: =
(\frac{1}{2},\frac{1}{2})$

\item[]
$\sqrt{(a,b)}: = (b, \neg a)$.

\end{enumerate}

\noindent Then  $ \langle S_A, \oplus, \bullet, \neg, \sqrt{,} \hspace{0.2cm} 0, \frac{1}{2}, 1 \rangle$ is a $\sqrt
{qPMV}$-algebra, and for each pair of elements $(a,b), (c,d)$ in $S_A$, $(a,b) \leq (c,d)$ iff $a\leq c$ in $A$.

\end{prop}

\begin{proof}
It is not very hard to see that the  reduct $ \langle A \times A,
\oplus, \bullet, \neg, 0, \frac{1}{2}, 1  \rangle$ is a $\sqrt
{qMV}$-algebra. We only have to prove that $S_A$ satisfies axioms 6
and 7 of $\sqrt {qPMV}$-algebras.

Ax 6) $x \bullet (y \odot \neg z) = (x \bullet y) \odot \neg
(x\bullet z)$.  In fact, $(a,b) \bullet ((c,d) \odot \neg (z,w)) =
(a,b) \bullet ((c,d) \odot (\neg z, \neg w)= (a \bullet (c \odot
\neg z ), \frac{1}{2}) = ((a\bullet c) \odot \neg(a \bullet z),
\frac{1}{2})$. On the other hand $((a,b) \bullet (c,d))  \odot \neg
((a,b) \bullet (z,w)) = (a\bullet c, \frac{1}{2}) \odot
(\neg(a\bullet z), \frac{1}{2}) = ((a\bullet c) \odot \neg(a\bullet
z), \frac{1}{2} ) $.

Ax 7) $\sqrt{x\bullet y} \oplus 0 =  \frac{1}{2}$. In fact:
$\sqrt{(a,b)\bullet (c,d)} \oplus (0,\frac{1}{2}) = \sqrt{(a\bullet
c, \frac{1}{2})} \oplus (0, \frac{1}{2}) = (\frac{1}{2}, \neg
(a\bullet c)) \oplus (0, \frac{1}{2})= (\frac{1}{2}, \frac{1}{2})= \frac{1}{2}$.\\

\noindent Hence $S_A$ is a $\sqrt{qPMV}$-algebra. Therefore we have
that $(a,b) \leq (c,d)$ iff $(1, \frac{1}{2}) = (a,b) \rightarrow
(c,d) = (\neg a \oplus b, \frac{1}{2})$ iff $a\leq b$ in $A$.

\qed
\end{proof}

We denote by ${\cal S}^\Box$ the class of algebras $S_A$ built in Proposition \ref{ST} where $A$ is a $PMV$-chain.

\begin{prop}\label{PROJECTION}
Let $S_A$ be a ${\cal S}^\Box$-algebra from the $PMV$-chain $A$. Then $R(S_A)$ is ${\cal PMV}$-isomorphic to $A$.
\end{prop}

\begin{proof}
If we consider $S_A \oplus 0 = \{(x,y) \oplus (0,\frac{1}{2}):
(x,y)\in A \times A \}$ then we have that $S_A \oplus 0 =
\{(x,\frac{1}{2}): x\in A \}$. Therefore, $S_A \oplus 0$ is
${\cal PMV}$-isomorphic to $A$. Using Proposition \ref{REG1} we have that
$R(S_A)$ is ${\cal PMV}$-isomorphic to $A$.

\qed
\end{proof}

\begin{prop}\label{AUX1}
Let $A$ be a $\sqrt{qPMV}$-algebra and $t = t(x_1, \ldots, x_n)$ be a $\sqrt
{q{\cal PMV}}$-term.

\begin{enumerate}
\item
If $t$ contains a subterm of the form $s_1 \oplus s_2$ then,
for each $\bar a \in A^n $, $t^A[\bar a] \oplus 0 = 1$ implies
that $t^A[\bar a] = 1$.

\item
If $A$ is a sub algebra of a ${\cal S}^\Box$-algebra and $A
\models t=1$ then there exists a $\sqrt {q{\cal PMV}}$-term $t'$
such that $\sqrt {q{\cal PMV}} \models t=t' \oplus 0 $.
\end{enumerate}
\end{prop}

\begin{proof}
1) Induction on the complexity of $t$. Since $t$ contains at least
an occurrence of $\oplus$, it cannot be an atomic term. Its minimum
possible complexity is therefore represented by the case $t = s_1
\oplus s_2$ where each $s_i$ is either a variable or constant, and
our claim trivially follows form Axiom Q6. Now let our claim hold
whenever the complexity of a term is less than $n$, and let $t$ have
complexity $n$. If $t^A[\bar a] \in Reg(A)$ our claim trivially
follows. Suppose that $t^A[\bar a] \not \in Reg(A)$. Then $t \not =
s_1 \oplus s_2$ and $t \not = s_1 \bullet s_2$. By SQ3 we have to
consider the case $t = \sqrt{s}$. There are two possible subcases.
\hspace{0.2cm} a) If $s = s_1 \star s_2$ such that $\star \in
\{\oplus, \bullet\}$ then $t^A[\bar a] \oplus 0 = 1$ implies that $1
= \sqrt{(s^A_1 \star s^A_2)[\bar a]} \oplus 0 = \frac{1}{2}$. In
this case $0=1$ and $A$ is a trivial algebra. \hspace{0.2cm} b) $s =
\sqrt{s_1}$. Then $t^A[\bar a] \oplus 0 = 1$ implies that $1=
\sqrt{\sqrt{s^A_1[\bar a]}} \oplus 0 = \neg s^A_1[\bar a] \oplus 0$.
Since complexity of $\neg s_1$ is $n-1$ we
have that $1 = \neg s^A_1[\bar a] = \sqrt{\sqrt{s^A_1[\bar a]}} = t^A[\bar a]$. \\

2) Induction on the complexity of $t$ again. If $t$ is atomic then
$t=1$. Now let our claim hold whenever the complexity of a term is
less than $n$, and let $t$ have complexity $n$. If $t = t_1 \star
t_2$ such that $\star \in \{ \oplus, \bullet \}$ we can consider $t'
= t$. Suppose that $t = \sqrt{s}$. In this case $A \models t=1$ iff
for any vector $\bar a$ in $A$, $s^A[\bar a] = (\frac{1}{2}, 1)$. It
is clear that $s \not = s_1 \star s_2$ with $\star \in \{\oplus,
\bullet \}$ since $s_1 \star s_2$ has the form $(a,\frac{1}{2})$ in
$A$. By SQ3 we have to suppose that $s= \sqrt{s_1}$. In this case $t =
\sqrt{\sqrt{s_1}}$. Therefore $A \models t=1$ implies that $A
\models \neg s_1 = 1$. Since $\neg s_1$ has a complexity $n-1$, then
there exists a $\sqrt {q{\cal PMV}}$-term $t'$ such that $\sqrt
{q{\cal PMV}} \models \neg s_1=t' \oplus 0$. Thus $\sqrt {q{\cal PMV}} \models t = t' \oplus 0 $.

\qed
\end{proof}

\section{The irreversible Poincar\'{e} structure}
In this section we will introduce the algebraic framework for the
Poincar\'{e} irreversible quantum computational system. In the
precedent section we have seen that the  $\sqrt{qPMV}$-structure
captures the basic properties of the $IP$-algebra $ \langle {\cal
D}({\mathbb{C}}^2), \bullet, \oplus, \neg, \sqrt, \hspace{0.1cm}
P_0, \rho_\frac{1}{2}, P_1  \rangle$ but it is not able to express
in an abstract form the relation between density operators $\sigma =
(r_1, r_2, r_3)$ and $\sqrt{\sigma}$ given in Lemma
\ref{PROBBLOC2}-10. This section is devoted to motivate and
construct a structure able to capture the mentioned items of  Lemma
\ref{PROBBLOC2}.

\subsection{Irreversible Poincar\'{e} structure in the plane}
The relation between $\sigma = (r_1, r_2, r_3)$ and $\sqrt{\sigma}$
with respect to the probability values they may take depends on the
relation between the components $r_2,r_3$ given in Lemma
\ref{PROBBLOC2}-8. This fact suggests the analysis of an abstraction
of the  $IP$-algebra restricted to the  $Y-Z$ plane.

\begin{lem} \label{REDUCE1}
${\cal D}({\mathbb{C}}^2)_{y,z} = \{\sigma = (0, r_2, r_3): \sigma
\in {\cal D}({\mathbb{C}}^2)\}$ is a sub universe of ${\cal
D}({\mathbb{C}}^2)$ resulting a sub $\sqrt{qPMV}$-algebra of ${\cal
D}({\mathbb{C}}^2)$. Moreover for each $\sqrt {q{\cal PMV}}$-term $t$
$${\cal D}({\mathbb{C}}^2)_{y,z}\models t = 1 \hspace{0.3cm} iff
\hspace{0.3cm} {\cal D}({\mathbb{C}}^2)\models t = 1 $$
\end{lem}

\begin{proof}
By definition of $\oplus$ and $\bullet$ it is clear that both are
closed operations in ${\cal D}({\mathbb{C}}^2)_{y,z}$. By Lemma
\ref{PROBBLOC2}-8, $\neg$ and $\sqrt{,}$ are also closed in ${\cal D}({\mathbb{C}}^2)_{y,z}$.\\

{\it Claim}. For each ${\cal D}({\mathbb{C}}^2)$-valuation $v: Term
\rightarrow {\cal D}({\mathbb{C}}^2)$ there exists a ${\cal
D}({\mathbb{C}}^2)_{y,z}$-valuation $v': Term \rightarrow {\cal
D}({\mathbb{C}}^2)$ such that $v(t) \oplus P_0 = v'(t) \oplus P_0$ .
For the constant terms, $v$ and $v'$ must coincide. If $t$ is a
variable such that $v(t) = (r_1, r_2, r_3)$, we define $v'(t) = (0,
r_2, r_3)$. Therefore,  by Lemma \ref{PROBBLOC1} $v(t) \oplus P_0 =
\rho_{\frac{1-r_3}{2}} = v'(t) \oplus P_0$. In the usual way we can
extend $v'$ to the set $Term$. Now we use induction. If $t$ is
$t_1\star t_2$ such that $\star \in \{\oplus, \bullet \}$, taking
into account Proposition \ref{CHIQPMV}-5, $v'(t) \oplus P_0 =
v'(t_1\star t_2) \oplus P_0 = (v'(t_1)\oplus P_0) \star
(v'(t_2)\oplus P_0) = (v(t_1)\oplus P_0) \star (v(t_2)\oplus P_0) =
v(t_1\star t_2) \oplus P_0 = v(t) \oplus P_0$. If $t$ is $\neg s$,
it follows from axiom Q5. If $t$ is $\sqrt{s}$, we must consider
three cases:

{\it Case 1}: $s$ is a variable such that $v(s) = (r_1, r_2, r_3)$.
By Lemma \ref{PROBBLOC2}-8 $v(t) \oplus P_0 = \sqrt{v(s)} \oplus P_0
= (r_1, -r_3, r_2) \oplus P_0 = \rho_{\frac{1-r_2}{2}} = (0, -r_3,
r_2) \oplus P_0 = \sqrt{v'(s)} \oplus P_0 = v'(t) \oplus P_0$.

{\it Case 2}: $s$ is $s_1\star s_2$. By Lemma \ref{PROBBLOC2}-9,
$v(t) \oplus P_0 = \sqrt{v(s_1\star s_2)} \oplus P_0 =
\rho_{p(\sqrt{v(s_1\star s_2)})} = \rho_{\frac{1}{2}} =
\sqrt{v'(s_1\star s_2)} \oplus P_0 = v'(t) \oplus P_0$

{\it Case 3}: $s$ is $\neg s_1$ or $\sqrt{s_1}$  it is already included in
the previous cases in view of Lemma \ref{PROBBLOC2}-6 and 7. Thus, $v(t) \oplus P_0 = v'(t) \oplus P_0$ as is required.\\

Assume that ${\cal D}({\mathbb{C}}^2)_{y,z}\models t = 1$. Let $v:
Term \rightarrow {\cal D}({\mathbb{C}}^2)$ be a valuation. By the
claim, there exists a valuation $v': Term \rightarrow {\cal
D}({\mathbb{C}}^2)_{y,z}$ such that $v'(t) \oplus P_0 = v(t) \oplus
P_0$ and clearly $v'(t) = P_1 = (0,0,-1)$. Hence $P_1 = v(t) \oplus
P_0 = \rho_{p(v(t))} = (0,0,-1)$ iff $v(t) = (0,0,-1)$.

\qed
\end{proof}

Let $S_{[0,1]}$ be the ${\cal S}^\Box$-algebra from
$[0,1]_{PMV}$. If we consider the set $$D_{[0,1]} = \{(x,y)\in
S_{[0,1]} : (x - \frac{1}{2})^2 + (y - \frac{1}{2})^2 \leq
\frac{1}{4} \}$$ it is not very hard to see that $D_{[0,1]}$ is a
sub universe of $S_{[0,1]}$. Thus $\langle D_{[0,1]}, \oplus,
\bullet, \neg, \sqrt{,} \hspace{0.2cm}, 0,\frac{1}{2}, 1 \rangle$ is
a sub $\sqrt{qPMV}$-algebra of $S_{[0,1]}$. In {\rm \cite{LKPG1}} is
proved that the reduct $\langle D_{[0,1]}, \oplus, \neg, 0, 1
\rangle$ characterize the equational theory of the $q{\cal MV}$.

\begin{lem} \label{LAGRANGE}
Let $(x,y) \in S_{[0,1]}$. Then we have:

\begin{enumerate}

\item
$\frac{x}{4} \oplus  \frac{y}{4} \leq \frac{ 1 + \sqrt 2}{4 \sqrt
2}$ \hspace{0,2cm} iff \hspace{0,2cm} $(x,y) \in D_{[0,1]} $,

\item
$\frac{x}{4} \oplus \frac{1}{8} \leq \frac{3}{8} \leq \frac{ 1 +
\sqrt 2}{4 \sqrt 2}$.

\item
$\varphi: {\cal D}({\mathbb{C}}^2)_{y,z} \rightarrow D_{[0,1]}$ such that
$\varphi(y,z) = (\frac{1-z}{2}, \frac{1-y}{2})$ is a
$\sqrt{q{\cal PMV}}$-isomorphism.

\end{enumerate}
\end{lem}

\begin{proof}
1) We first note that that $\frac{x}{4} \oplus  \frac{y}{4} =
\frac{x}{4} +  \frac{y}{4} $. Consider the function $\frac{x}{4} +
\frac{y}{4}$ subject to the constraint  $(x-\frac{1}{2})^2 +
(y-\frac{1}{2}) ^2 \leq \frac{1}{4}$. Using Lagrange multipliers, we
obtain the following equation system

$$\nabla f = \nabla (\frac{x}{4} +  \frac{y}{4})  = \lambda \nabla [(x-\frac{1}{2})^2 + (y-\frac{1}{2})^2 -
1]$$ $$(x-1/2)^2 + (y-1/2)^2 = 1/4$$

\noindent It is equivalent to the system $$\{ \hspace{0.2cm}
2\lambda (x-\frac{1}{2}) = \frac{1}{4} , \hspace{0.3cm}  2\lambda
(y-\frac{1}{2}) = \frac{1}{4} , \hspace{0.3cm} (x-\frac{1}{2})^2 +
(y-\frac{1}{2})^2 = \frac{1}{4} \hspace{0.2cm} \} $$

\noindent and it is not very hard to see that $x = y = \frac{1}{2} +
\frac{1}{2\sqrt2}$ is a solution of this system, giving a maximum of
$\frac{x}{4} +  \frac{y}{4} $ in the mentioned restriction . Thus
$\frac{x}{4} \oplus  \frac{y}{4} \leq \frac{ 1 + \sqrt 2}{4 \sqrt
2}$. To see the converse, assume that $\frac{x}{4} \oplus
\frac{y}{4} \leq \frac{ 1 + \sqrt 2}{4 \sqrt 2}$. Let $x =
\frac{1}{2} + r\cos \theta$ and $y = \frac{1}{2} + r\sin \theta$,
therefore from $\frac{1}{4}(\frac{1}{2} + r\cos \theta) +
\frac{1}{4}(\frac{1}{2} + r\sin \theta) \leq \frac{ 1 + \sqrt 2}{4
\sqrt 2} $ we have that $r(\cos \theta + \sin \theta) \leq
\frac{1}{\sqrt 2}$.   But the maximum of $(\cos \theta + \sin
\theta)$ is given when $\theta = \frac{\pi}{4} $. In this case
$r\frac{2}{\sqrt 2} \leq \frac{1}{\sqrt 2}$ resulting
$r \leq \frac{1}{ 2}$ \\

2) Immediate. \\

3) Let $ \sigma =(0, b,c) \in {\cal D}({\mathbb{C}}^2)_{y,z}$. Then
$\varphi(\sigma) = (\frac{1-c}{2}, \frac{1-b}{2}) $ and
$(\frac{1-c}{2} - \frac{1}{2})^2 + (\frac{1-b}{2} - \frac{1}{2})^2 =
\frac{1}{4}(c^2 + b^2) \leq \frac{1}{4}$. Thus the image of
$\varphi$ is contained in $D_{[0,1]}$. It is clear that $\varphi$ is
injective. Let $(a,b) \in D_{[0,1]}$. If we consider $\sigma = (0,
1-2b, 1- 2a)$ then $(1-2b)^2 + (1- 2a)^2 = 4(\frac{1}{2} - a)^2 +
4(\frac{1}{2} - b)^2 \leq 1$. Hence $\sigma \in {\cal
D}({\mathbb{C}}^2)_{y,z}$, $\varphi(\sigma) = (a,b)$ and $\varphi$
is a surjective map. Now we prove that $\varphi$ is a
$\sqrt{q{\cal PMV}}$-homomorphism. Let $\sigma =(0, r_2, r_3)$ and $\tau
=(0, s_2, s_3)$. Using Lemma \ref{PROBBLOC1} and Lemma
\ref{PROBBLOC2} \ref{PROBBLOC2} we have that:

\begin{itemize}

\item
Let $\star \in \{\oplus, \bullet \}$. $\varphi(\sigma \star \tau) =
\varphi (\rho_{p(\sigma)\star p(\rho)}) = \varphi (0, 0, 1-
2(p(\sigma) \star p(\rho))) = (p(\sigma)\star p(\rho), \frac{1}{2})
= (\frac{1-r_3}{2}, \frac{1-r_2}{2}) \star (\frac{1-s_3}{2},
\frac{1-s_2}{2}) = \varphi(\sigma) \star \varphi(\tau)$.

\item
$\varphi(\sqrt{\sigma}) = \varphi(0, -r_3, r_2) = (\frac{1-r_2}{2},
\frac{1+r_3}{2}) = (\frac{1-r_2}{2}, 1- \frac{1-r_3}{2}) =
\sqrt{(\frac{1-r_3}{2},  \frac{1-r_2}{2})} = \sqrt{\varphi
(\sigma)}$.

\item
$\varphi(P_1) = \varphi(0,0,-1) = (1, \frac{1}{2})$, \hspace{0.2cm}
$\varphi(P_0) = \varphi(0,0,1) = (0, \frac{1}{2})$ \hspace{0.2cm}
and \hspace{0.2cm} $\varphi(\rho_{\frac{1}{2}}) = \varphi(0,0,0) =
(\frac{1}{2}, \frac{1}{2})$.

\end{itemize}

\noindent
Thus $\varphi$ is $\sqrt{q{\cal PMV}}$-isomorphism.

\qed
\end{proof}

In view of Lemma \ref{REDUCE1} and Lemma \ref{LAGRANGE} we can establish the following proposition:

\begin{theo}\label{REDUCE3}
For each $\sqrt {q{\cal PMV}}$-term $t$ we have that $${\cal D}({\mathbb{C}}^2)\models
t = 1 \hspace{0.3cm} iff \hspace{0.3cm} D_{[0,1]}\models t = 1
$$ \qed
\end{theo}

By Lemma \ref{LAGRANGE}, $D_{[0,1]}$ satisfies the relation between
$x$ and $\sqrt{x}$ claimed by Lemma \ref{PROBBLOC2} (items 10 and 11).
Since for our logical system only $\sqrt {q{\cal PMV}}$-equations of the
form $t=1$ are need, in view of Theorem \ref{REDUCE3}, $D_{[0,1]}$
is a more appropriate standard frame than $S_{[0,1]}$ for the
algebra of quantum gates.

\subsection{$PMV$-algebras with fix point of the negation}
To obtain an algebraic structure able to generalize  $D_{[0,1]}$ it
is necessary to represent the inequality  $\frac{x}{4} \oplus
\frac{\sqrt x}{4} \leq \frac{ 1 + \sqrt 2}{4 \sqrt 2}$ given in
Lemma \ref{LAGRANGE}. Taking into account Lemma \ref{DENS}, we have
that: $$\frac{x}{4} \oplus \frac{\sqrt x}{4} \leq \frac{ 1 + \sqrt 2}{4
\sqrt 2} \Longleftrightarrow \forall s \in G_{[0,1]}(\frac{1}{2})
\hspace{0.2cm} s.t. \hspace{0.2cm} s \geq \frac{ 1 + \sqrt 2}{4
\sqrt 2}, \hspace{0.3cm} \frac{x}{4} \oplus \frac{\sqrt x}{4} \leq s
$$ In terms  of the language of $\sqrt {q{\cal PMV}}$, the second part of
the above equivalence can be expressed through the following set of
equations: $$\{1 = ( (\frac{1}{4}\bullet x)  \oplus
(\frac{1}{4}\bullet \sqrt x) ) \rightarrow s: s \in
G_{[0,1]}(\frac{1}{2}) \hspace{0.2cm} and \hspace{0.2cm} s \geq
\frac{1+\sqrt2}{4\sqrt2} \}$$

To represent this set of equations we need consider a
subclass of the $\sqrt{qPMV}$-algebras such that their regular elements have an isomorphic copy of $G_{[0,1]}(\frac{1}{2})$.
In this subsection we give an equational theory for the class $PMV$-algebras containing an isomorphic copy of $G_{[0,1]}(\frac{1}{2})$
as sub $PMV$-algebra. \\

It is well known that a $PMV$-algebra has at most a fix point of the
negation \cite[Lemma 2.10]{UHO}. In $[0,1]_{PMV}$ we have that $\neg
\frac{1}{2} = \frac{1}{2}$. In the type of algebras $\langle \oplus,
\bullet, \neg, 0, \frac{{\bf 1}}{{\bf 2}}, 1 \rangle$,  $$
\frac{{\bf 1}}{{\bf 2}^n} \hspace{0.1cm} design \hspace{0.1cm}
\cases {\frac{{\bf 1}}{{\bf 2}}, & if $n=1$  \cr (\frac{{\bf
1}}{{\bf 2}^{n-1}}) \bullet \frac{{\bf 1}}{{\bf 2}}, & if $n > 1
$\cr } $$ and for each $a \in N$ such that $0 \leq a \leq 2^n$,

$$
\frac{a}{{\bf 2}^n} = \bigoplus_a \frac{{\bf 1}}{{\bf 2}^n} \hspace{0.1cm} design
\hspace{0.1cm}  \cases {0, & if $a=0$  \cr  (\bigoplus_{a-1}
\frac{{\bf 1}}{{\bf 2}^n}) \oplus \frac{{\bf 1}}{{\bf 2}^n}, & if
$1 \leq a \leq 2^n $\cr }
$$

\noindent  In that follows $\frac{{\bf 1}}{{\bf 4}}$
design the term  $\frac{{\bf 1}}{{\bf 2}^2}$ and $\frac{{\bf
1}}{{\bf 8}}$ design the term  $\frac{{\bf 1}}{{\bf 2}^3}$.

\begin{definition}
{\rm A $PMV_\frac{1}{2}$-algebra is an algebra $\langle A, \oplus, \bullet , \neg, 0, {\bf \frac{1}{2}}, 1\rangle$ of type
$ \langle 2, 2, 1, 0, 0,0 \rangle$ satisfying the following:

\begin{enumerate}
\item[1]
$\langle A, \oplus, \bullet , \neg, 0, {\bf \frac{1}{2}}, 1\rangle$ is a $PMV$-algebra,

\item[2]
$\neg {\bf \frac{1}{2}} = {\bf \frac{1}{2}}$,

\item[3]
$\frac{a}{{\bf 2}^n} \odot \frac{b}{{\bf 2}^m} =  \frac{\max \{0, \hspace{0.1cm} a+b2^{n-m}-2^n\}}{{\bf 2}^n}$ with $n \geq m$,

\item[4]
$\frac{a}{{\bf 2}^n} \bullet \frac{b}{{\bf 2}^m} = \frac{ab}{{\bf 2}^{n+m}} $,

\item[5]
$\neg \frac{a}{{\bf 2}^n} = \frac{2^n - a}{{\bf 2}^n}$.

\end{enumerate}
}
\end{definition}

We denote by ${\cal PMV}_\frac{1}{2}$ the variety of
$PMV_\frac{1}{2}$-algebras. It is clear that $[0,1]_{PMV}$ is a
$PMV_\frac{1}{2}$-algebra. If $A$ is a $PMV_\frac{1}{2}$-algebra
then $G_A({\bf \frac{1}{2}})$ design the sub algebra of $A$ generated by
$\{0, {\bf \frac{1}{2}}, 1\}$. By Axiom 3,4,5  and  induction in the
complexity of terms we can establish the following lemma:

\begin{lem} \label{MEDIOREP}
Let $A$ be a $PMV_\frac{1}{2}$-algebra. Then for each $x \in G_A({\bf \frac{1}{2}})$, $x = \frac{a}{{\bf 2}^n}$ for some $a \leq 2^n$. \qed
\end{lem}

\begin{theo} \label{MEDIOREP1}
Let $A$ be a $PMV_\frac{1}{2}$-algebra. Then there exists an unique
${\cal PMV}_{\frac{1}{2}}$-isomorphisms from $G_A({\bf \frac{1}{2}})$ onto
$G_{[0,1]}(\frac{1}{2})$.
\end{theo}

\begin{proof}
We first prove that $G_A({\bf \frac{1}{2}})$ is a simple algebra. By
Proposition \ref{SIMPMV} and Proposition \ref{CON}-2 we have to see
that for each $x\not=1$ in $G_A({\bf \frac{1}{2}})$, $x$ is
$\odot$-nilpotent. Suppose that $x = \frac{a}{{\bf 2}^n}$. By
induction, we can see that  $\bigodot_k \frac{a}{{\bf 2}^n} =
\frac{\max \{0, \hspace{0.1cm} ka -(k-1)2^n \}}{{\bf 2}^n}$. Thus
for $k \geq \frac{2^n}{2^n - a}$ result $\bigodot_k \frac{a}{{\bf
2}^n} = 0$ and then $x$ is $\odot$-nilpotent. Hence $G_A({\bf
\frac{1}{2}})$ is simple. By Proposition \ref{CON}-3 if we consider
the following ${\cal PMV}$-homomorphism ${\bf 2} \rightarrow [0,1]_{PMV}$
and ${\bf 2} \rightarrow G_A({\bf \frac{1}{2}})$ there exists a
${\cal PMV}$-homomorphism $f:G_A({\bf \frac{1}{2}}) \rightarrow
[0,1]_{PMV}$ such that the following diagram is commutative:

\begin{center}
\unitlength=1mm
\begin{picture}(20,20)(0,0)
\put(8,16){\vector(3,0){5}} \put(2,10){\vector(0,-2){5}}
\put(10,4){\vector(1,1){7}}

\put(2,10){\makebox(13,0){$\equiv$}}

\put(2,16){\makebox(0,0){${\bf 2}$}}
\put(22,16){\makebox(0,0){$[0,1]_{PMV}$}}
 \put(2,0){\makebox(0,0){$G_A({{\bf \frac{1}{2}}})$}}
\put(18,2){\makebox(-4,3){$f$}}
\end{picture}
\end{center}

Since ${\bf \frac{1}{2}}$ is the unique fix point of the negation in
$G_A({\bf \frac{1}{2}})$ it is clear that $f({\bf \frac{1}{2}}) =
\frac{1}{2}$. Thus $f$ is a ${\cal PMV}_\frac{1}{2}$-homomorphism. Since
$G_A({\bf \frac{1}{2}})$ is simple then $f$ is injective. It is
clear that $Imag(f) \subseteq G_{[0,1]}(\frac{1}{2})$. We prove that
$Imag(f) = G_{[0,1]}(\frac{1}{2})$. Let $x\in
G_{[0,1]}(\frac{1}{2})$. Then there exists a $PMV_\frac{1}{2}$-term
$t$ such that $x = t^{G_{[0,1]}(\frac{1}{2})}[0, \frac{1}{2}, 1]$.
Since $f$ is a ${\cal PMV}_\frac{1}{2}$-homomorphism then we have
$f(t^{G_A({\bf \frac{1}{2}})}[0, {\bf \frac{1}{2}}, 1]) =
t^{G_{[0,1]}(\frac{1}{2})}[f(0), f({\bf \frac{1}{2}}), f(1)]
=t^{G_{[0,1]}(\frac{1}{2})}[0, \frac{1}{2}, 1] = x$ and $Imag(f) =
G_{[0,1]}(\frac{1}{2})$. Hence $f$ is a
${\cal PMV}_\frac{1}{2}$-isomorphism form  $G_A(\frac{1}{2})$ onto
$G_{[0,1]}(\frac{1}{2})$. If $f'$ is other
${\cal PMV}_\frac{1}{2}$-isomorphisms then $f$ and $f'$ coincides over
$\{0,\frac{1}{2}, 1\}$. Therefor, by induction on the complexity of
terms it follows that $f = f'$.

\qed
\end{proof}

\begin{rema}\label{CONVENC}
{\rm From the last theorem whenever $A$ is a
$PMV_\frac{1}{2}$-algebra, we will have the following identification
$G_A(\frac{1}{2}) = G_{[0,1]}(\frac{1}{2})$. Thus each $s \in
G_A(\frac{1}{2})$ is seen as a unique element of
$G_{[0,1]}(\frac{1}{2})$. }
\end{rema}

\subsection{Irreversible Poincar\'{e} algebras}
In view of Theorem \ref{MEDIOREP1} now we can introduce a
substructure of $\sqrt{q{\cal PMV}}$  that allows to capture the
relation between  $x$ and $\sqrt{x}$ given in Lemma
\ref{PROBBLOC2}-10.

\begin{definition}
{\rm An {\it irreversible Poincar\'{e} algebra} is a
$\sqrt{qPMV}$-algebra satisfying the following axioms:

\begin{enumerate}
\item[P1]
$R(A)$ is a $PMV_\frac{1}{2}$-algebra,

\item[P2]
$1 = ( (\frac{1}{4}\bullet x)  \oplus  (\frac{1}{4}\bullet \sqrt x)
) \rightarrow s $ \hspace{0.3cm} where $s\in G_{R(A)}(\frac{1}{2})$
and $s \geq \frac{1+\sqrt2}{4\sqrt2}$.

\end{enumerate}

}
\end{definition}

It is clear that the Poincar\'{e} structure conforms a variety since
Axiom P1 is satisfied by adding Ax 3, Ax 4 and Ax 5 of ${\cal
PMV}_\frac{1}{2}$ to $\sqrt{q{\cal PMV}}$. We denote by ${\cal IP}$
the subvariety of $\sqrt{q{\cal PMV}}$ conformed by the irreversible
Poincar\'{e} algebras.

\begin{rema}
{\rm In view of Lemma \ref{LAGRANGE} it is not hard to see  that
$D_{[0,1]}$ is a ${\cal IP}$-algebra, being the
``standard  model" of ${\cal IP}$. Unfortunately we
cannot give a completeness theorem for the ${\cal IP}$-equations
of the form $t=1$ with respect to  $D_{[0,1]}$. In fact, the open
problem of axiomatization of all identities in the language of
${\cal PMV}$ which are valid in the $PMV$-algebra arising from the
real interval $[0,1]$ (see {\rm \cite{MONT, IS}}) will appear in
${\cal IP}$. In view of this, we delineate a
generalization of the $D_{[0,1]}$ algebra, whose role is analogous
to the  $PMV$-chains with  respect to the equational theory of ${\cal PMV}$. }
\end{rema}

Let $S_A$ be a ${\cal S}^\Box$-algebra from the
$PMV_{\frac{1}{2}}$-chain $A$. In view of Remark \ref{CONVENC}, for
$(a,b) \in S_A$ the expression $(a,b) \leq \frac{1+\sqrt2}{4\sqrt2}$
should be understood as $(a,b) \leq (s, \frac{1}{2})$ for all $s \in
G_A({\bf \frac{1}{2}})$ such that $s \geq \frac{1+\sqrt2}{4\sqrt2}$
(or equivalent $a \leq s$ in $A$ for all $s \in G_A({\bf
\frac{1}{2}})$). We consider the following partition in $S_A$:

\begin{enumerate}
\item[]
Quadrant I:   $ A^\llcorner =  \{(x,y) \in A^2: x \geq \frac{1}{2} ,
y \geq \frac{1}{2} \} $ $$ Q1= \{(x,y) \in A^\llcorner:
(\frac{1}{4}\bullet x) \oplus  (\frac{1}{4}\bullet y) ) \leq
\frac{1+\sqrt2}{4\sqrt2} \} $$

\item[]
Quadrant 2: $A^\lrcorner =  \{(x,y) \in A^2: x \leq \frac{1}{2}, y
\geq \frac{1}{2} \} $ $$ Q2= \{(x,y) \in A^\lrcorner:
(\frac{1}{4}\bullet \neg x)  \oplus  (\frac{1}{4}\bullet y) ) \leq
\frac{1+\sqrt2}{4\sqrt2} \} $$

\item[]
Quadrant 3:  $ A^\urcorner = \{(x,y) \in A^2: x \leq \frac{1}{2}, y
\leq \frac{1}{2} \} $ $$ Q3= \{(x,y) \in A^\urcorner:
(\frac{1}{4}\bullet \neg x)  \oplus  (\frac{1}{4}\bullet \neg y) )
\leq  \frac{1+\sqrt2}{4\sqrt2} \} $$

\item[]
Quadrant 4: $A^\ulcorner = \{(x,y) \in A^2: x \geq \frac{1}{2}, y
\leq \frac{1}{2} \} $  $$ Q4= \{(x,y) \in A^\ulcorner:
(\frac{1}{4}\bullet x)  \oplus  (\frac{1}{4}\bullet \neg y) ) \leq
\frac{1+\sqrt2}{4\sqrt2} \} $$

\end{enumerate}
\noindent Then we define $$ D_A = Q1 \cup Q2 \cup Q3 \cup Q4 $$
Since $A$ is a chain it is clear that $A^2 = A^\llcorner \cup
A^\lrcorner \cup A^\urcorner \cup A^\ulcorner$.

\begin{prop} \label{ALGCIRC}
Let $S_A$ be a ${\cal S}^\Box$-algebra from the $PMV_{\frac{1}{2}}$-chain $A$. Then $D_A$ is a sub-universe of $S_A$
and $$\langle D_A, \oplus, \bullet, \sqrt, \hspace{0.2cm}
(0,\frac{1}{2}), (\frac{1}{2},\frac{1}{2}), (1,\frac{1}{2})  \rangle
$$  is the largest irreversible Poincar\'{e} algebra contained in
$S_A$. Moreover $Reg(D_A)$ is ${\cal PMV}$-isomorphic to $A$.

\end{prop}

\begin{proof}
We first note that $(x, \frac{1}{2}) \in D_A$ for each $x\in A$
since $(\frac{1}{4}\bullet x) \oplus (\frac{1}{4}\bullet
\frac{1}{2}) \leq \frac{1}{4} \oplus\frac{1}{8} = \frac{3}{8} \leq
\frac{1+\sqrt2}{4\sqrt2}$. Thus if $(x_1,y_1)$ and  $(x_2,y_2)$ are
in $D_A$ and $\star \in \{\oplus, \bullet \}$ then  $(x_1,y_1)\star
(x_2,y_2) = (x_1 \star x_2, \frac{1}{2}) \in D_A$. Hence $D_A$ is
closed by $\oplus$ and $\bullet$.
$D_A$ is closed by $\neg$. In fact:\\

If $(x,y) \in Q1$ then $\neg (x,y) = (\neg x, \neg y) \in Q3$ since

$(\frac{1}{4}\bullet \neg (\neg x))  \oplus  (\frac{1}{4}\bullet
\neg (\neg y)) = (\frac{1}{4}\bullet x)  \oplus  (\frac{1}{4}\bullet
y) \leq  \frac{1+\sqrt2}{4\sqrt2}$.\\

If $(x,y) \in Q2$ then $\neg (x,y) = (\neg x, \neg y) \in Q4$ since

$(\frac{1}{4}\bullet \neg x)  \oplus  (\frac{1}{4}\bullet \neg (\neg
y)) = (\frac{1}{4}\bullet \neg x)  \oplus  (\frac{1}{4}\bullet y)
\leq  \frac{1+\sqrt2}{4\sqrt2}$. \\

If $(x,y) \in Q3$ then $\neg (x,y) = (\neg x, \neg y) \in Q1$ since

$(\frac{1}{4}\bullet \neg x)  \oplus  (\frac{1}{4}\bullet \neg  y)
\leq  \frac{1+\sqrt2}{4\sqrt2}$.\\

If $(x,y) \in Q4$ then $\neg (x,y) = (\neg x, \neg y) \in Q2$ since

$(\frac{1}{4}\bullet \neg (\neg x))  \oplus  (\frac{1}{4}\bullet
\neg y) = (\frac{1}{4}\bullet x)  \oplus  (\frac{1}{4}\bullet \neg
y) \leq  \frac{1+\sqrt2}{4\sqrt2}$.\\

\noindent $D_A$ is closed by $\sqrt,$. In fact:

If $(x,y) \in Q1$ then $\sqrt{(x,y)} = (y, \neg x) \in Q4$ since

$(\frac{1}{4}\bullet y))  \oplus  (\frac{1}{4}\bullet \neg(\neg x))
= (\frac{1}{4}\bullet x)  \oplus  (\frac{1}{4}\bullet  y) \leq
\frac{1+\sqrt2}{4\sqrt2}$. \\

If $(x,y) \in Q2$ then $\sqrt{(x,y)} = (y, \neg x) \in Q1$ since

$(\frac{1}{4}\bullet y))  \oplus  (\frac{1}{4}\bullet \neg x) =
(\frac{1}{4}\bullet \neg x)  \oplus  (\frac{1}{4}\bullet  y) \leq
\frac{1+\sqrt2}{4\sqrt2}$. \\

If $(x,y) \in Q3$ then $\sqrt{(x,y)} = (y, \neg x) \in Q2$ since

$(\frac{1}{4}\bullet \neg y)  \oplus  (\frac{1}{4}\bullet (\neg x))
= (\frac{1}{4}\bullet \neg x)  \oplus  (\frac{1}{4}\bullet  \neg y)
\leq  \frac{1+\sqrt2}{4\sqrt2}$. \\

If $(x,y) \in Q4$ then $\sqrt{(x,y)} = (y, \neg x) \in Q3$ since

$(\frac{1}{4}\bullet \neg y)  \oplus  (\frac{1}{4}\bullet \neg (\neg
x)) = (\frac{1}{4}\bullet x)  \oplus  (\frac{1}{4}\bullet  \neg y)
\leq  \frac{1+\sqrt2}{4\sqrt2}$.\\

\noindent Thus $D_A$ is a subalgebra of $S_A$. We will see that
$D_A$ is an irreversible Poincar\'{e} algebra. If $(x,y) \in D_A$
then $(\frac{1}{4}\bullet (x,y)) \oplus (\frac{1}{4}\bullet
\sqrt{(x,y)}) = ((\frac{1}{4}, \frac{1}{2}) \bullet (x,y)) \oplus
((\frac{1}{4}, \frac{1}{2}) \bullet (y, \neg x)) =
((\frac{1}{4}\bullet x) \oplus (\frac{1}{4}\bullet y),
\frac{1}{2})$. Therefore we need to prove that $(\frac{1}{4}\bullet
x) \oplus (\frac{1}{4}\bullet y) \leq \frac{1+\sqrt2}{4\sqrt2}$ in
the $PMV$-algebra $A$. In fact: if $(x,y) \in Q1$ then the inequality
is valid. If $(x,y) \in Q2$ then $x \leq \frac{1}{2} \leq \neg x$.
Therefore $(\frac{1}{4}\bullet x) \oplus(\frac{1}{4}\bullet y) \leq
(\frac{1}{4}\bullet \neg x) \oplus (\frac{1}{4}\bullet y) \leq
\frac{1+\sqrt2}{4\sqrt2}$. If $(x,y) \in Q3$ then $x \leq
\frac{1}{2} \leq \neg x$ and $y \leq \frac{1}{2} \leq \neg y$.
Therefore $(\frac{1}{4}\bullet x) \oplus(\frac{1}{4}\bullet y) \leq
(\frac{1}{4}\bullet \neg x) \oplus (\frac{1}{4}\bullet \neg y) \leq
\frac{1+\sqrt2}{4\sqrt2}$. If $(x,y) \in Q4$ then $y \leq
\frac{1}{2} \leq \neg y$. Therefore $(\frac{1}{4}\bullet x)
\oplus(\frac{1}{4}\bullet y) \leq (\frac{1}{4}\bullet x) \oplus
(\frac{1}{4}\bullet \neg y) \leq    \frac{1+\sqrt2}{4\sqrt2}$.

Now we prove that $D_A$ is the largest irreversible Poincar\'{e}
algebra contained in $S_A$. Let $B$ be an irreversible Poincar\'{e}
algebra algebra contained in $S_A$ and let $(x,y)\in B$.

\begin{enumerate}
\item
If $x\geq \frac{1}{2}$ and $y\geq \frac{1}{2}$ then $(x,y)\in Q1$.

\item
Suppose that $x\leq \frac{1}{2}$ and $y\geq \frac{1}{2}$. Since
$\sqrt{(x,y)} = (y, \neg x) \in B$ then $\frac{1+\sqrt2}{4\sqrt2}
\geq (\frac{1}{4}\bullet \sqrt{(x,y)}) \oplus (\frac{1}{4}\bullet
\sqrt{ \sqrt{(x,y)})}) = ((\frac{1}{4}\bullet y) \oplus
(\frac{1}{4}\bullet \neg x) , \frac{1}{2})$. In this case
$(\frac{1}{4}\bullet y) \oplus (\frac{1}{4}\bullet \neg x) \leq
\frac{1+\sqrt2}{4\sqrt2}$ and $(x,y) \in Q2$.

\item
Suppose that $x\leq \frac{1}{2}$ and $y\leq \frac{1}{2}$. Since
$\neg (x,y) = (\neg x, \neg y) \in B$ then $\frac{1+\sqrt2}{4\sqrt2}
\geq (\frac{1}{4} \bullet \neg (x,y)) \oplus (\frac{1}{4} \bullet
\sqrt{ \neg (x,y))} = ( (\frac{1}{4}\bullet \neg x) \oplus
(\frac{1}{4}\bullet \neg y), \frac{1}{2})$. In this case
$(\frac{1}{4}\bullet \neg x) \oplus (\frac{1}{4}\bullet \neg y) \leq
\frac{1+\sqrt2}{4\sqrt2}$. Thus $(x,y) \in Q3$

\item
Suppose that $x\geq \frac{1}{2}$ and $y\leq \frac{1}{2}$. Since $
\neg \sqrt{(x,y)} = ( \neg y,  x) \in B$ then
$\frac{1+\sqrt2}{4\sqrt2} \geq  (\frac{1}{4}\bullet \neg
\sqrt{(x,y)}) \oplus (\frac{1}{4}\bullet \sqrt{ \neg \sqrt{(x,y)})})
= ((\frac{1}{4}\bullet \neg y) \oplus (\frac{1}{4}\bullet x) ,
\frac{1}{2})$. In this case $(\frac{1}{4}\bullet \neg y) \oplus
(\frac{1}{4}\bullet x) \leq \frac{1+\sqrt2}{4\sqrt2}$. Thus $(x,y)
\in Q4$.

\end{enumerate}

\noindent Thus $(x,y) \in D_A$ and $B$ is a Poincar\'{e} sub algebra
of $D_A$. With the same argument used in Proposition \ref{PROJECTION} we can prove that $Reg(D_A)$ is ${\cal PMV}$-isomorphic to $A$.

\qed
\end{proof}

We denote by ${\cal S}^\circ$ the class of algebras $D_A$ given in
Proposition \ref{ALGCIRC} where $A$ is a $PMV_{\frac{1}{2}}$-chain.
For the sake of simplicity in the notation, in the next theorem we
will use the following convention: ${\cal S}$ may be either ${\cal
S}^{\Box}$ or ${\cal S}^{\circ}$ and we define the class of algebras
${\cal A}_{{\cal S}}$ as follows:
$$
{\cal A}_{{\cal S}} = \cases {\sqrt{q{\cal PMV}}, & if ${\cal S} = {\cal S}^{\Box}$  \cr  {\cal IP}, & if  ${\cal S} = {\cal S}^{\circ} $\cr }
$$

\begin{prop}\label{REDUCE4}
For each $\sqrt{qPMV}$-term $t$ we have that $$ {\cal A}_{{\cal S}} \models t = 1 \hspace{0.4cm} iff \hspace{0.3cm} {\cal S} \models t = 1  $$

\end{prop}

\begin{proof}
Let $A$ be a $\sqrt{qPMV}$-algebra. We consider the $PMV$-algebra
$Reg(A)$ of all regular elements. By Proposition \ref{CON} we can
consider a subdirect representation $\beta:Reg(A)\hookrightarrow
 \Pi_{i\in I} A_i $ such that $(A_i)_{i\in I}$ is a family
of $PMV$-chains. If $x \in Reg(A)$ we write $\beta(x) = (x_i)_{i\in
I}$. Let $p_j$ the $j$-th projection $p_j: \Pi_{i\in I} A_i
\rightarrow A_j$. By Proposition \ref{ST}, we consider the
${\cal S}^{\Box}$-algebra $S_{Reg(A)}$ and for each $PMV$-chain $A_i$ we
consider the ${\cal S}^{\Box}$-algebra $S_{A_i}$ and $D_{A_i}$, the ${\cal S}^\circ$ sub
algebra of $S_{A_i}$. Define the function $$f:A
\rightarrow S_{Reg(A)} \hspace{0.5cm} s.t. \hspace{0.2cm}  x
\longmapsto (x\oplus 0, \sqrt{x} \oplus 0)$$ We need to prove that
$f$ is a $\sqrt{q{\cal PMV}}$-homomorphism.

\begin{itemize}

\item
Let $a \in \{0, \frac{1}{2}, 1 \}$. In this case $a\in Reg(A)$ and
$\sqrt{a} \oplus 0 = \frac{1}{2}$. Therefore $f(a) = (a \oplus 0,
\sqrt{a} \oplus 0) = (a, \frac{1}{2})$.

\item
Let $\star \in \{\oplus, \bullet \}$. $f(x \star y) = ((x \star y)
\oplus 0, \sqrt{x \star y} \oplus 0) = ((x \star y) \oplus 0,
\frac{1}{2}) = (x\oplus 0, \sqrt{x} \oplus 0) \star (y \oplus 0,
\sqrt{y} \oplus 0) = f(x) \star f(y)$.

\item
$f(\sqrt{x}) = (\sqrt{x} \oplus 0,  \sqrt{\sqrt{x}} \oplus 0) =
(\sqrt{x} \oplus 0, \neg (x \oplus 0) ) = \sqrt{(x \oplus 0,
\sqrt{x} \oplus 0)} = \sqrt{f(x)}$. Consequently $f(\neg x) = \neg
f(x)$ in view of Axiom SQ3.

\end{itemize}

\noindent Thus $f$ is a $\sqrt{q{\cal PMV}}$-homomorphism. For each $i\in
I$ we consider the function  $$\beta_i :S_{Reg(A)} \rightarrow
S_{A_i} \hspace{0.5cm} s.t. \hspace{0.2cm}  (x,y) \longmapsto (p_i
\beta(x), p_i \beta(y))_{i \in I}$$

We will prove that $\beta_i$ is a $\sqrt{q{\cal PMV}}$-homomorphism for
each $i \in I$. Let $(x,y), (x_1,y_1), (x_2,y_2) \in  S_{Reg(A)}$

\begin{itemize}

\item
The cases $(0,\frac{1}{2}), (\frac{1}{2}, \frac{1}{2}), (1, \frac{1}{2})$ are immediate.

\item
Let $\star \in \{\oplus, \bullet \}$. $\beta_i( (x_1,y_1) \star (x_2,y_2)) =
 \beta_i (x_1 \star x_2, \frac{1}{2}) = ({x_1}_i \star {x_2}_i, \frac{1}{2}_i) =
 ({x_1}_i ,{y_1}_i) \star ({x_2}_i,{y_2}_i) = \beta_i((x_1,y_1)) \star \beta_i((x_2,y_2)) $.

\item
$\beta_i(\sqrt{(x,y)}) = \beta_i((y, \neg x)) = (y_i, \neg x_i) =
\sqrt{(x_i,y_i)} = \sqrt{\beta_i(x,y)}$. Consequently $f(\neg x) =
\neg f(x)$ in view of Axiom SQ3.

\end{itemize}

Thus $\beta_i$ is a $\sqrt{q{\cal PMV}}$-homomorphism for each $i \in I$. Now we prove the theorem.

$\Longrightarrow$) Immediate

$\Longleftarrow$) Assume that ${\cal S} \models t = 1$. By
Proposition \ref{AUX1} we can identify $t$ with $t\oplus 0$. Suppose
that there exists $A \in {\cal A}_{{\cal S}}$ and $\bar a \in A^n$ such that $t^A[\bar a] \oplus 0 \not = 1$. Then $f(t^A[\bar a] \oplus 0) =
((t^A[\bar a] \oplus 0) \oplus 0, \sqrt{t^A[\bar a] \oplus 0} \oplus 0) =
(t^A[\bar a] \oplus 0, \frac{1}{2}) \not = (1, \frac{1}{2}) $ in
$S_{Reg(A)}$. It is clear that $t^A[\bar a] \oplus 0 \in Reg(A)$. By the
subdirect representation of the $PMV$-algebra $Reg(A)$, there exists
$A_i$ such that $(t^A[\bar a] \oplus 0)_i \not = 1_i$ in $A_i$. Therefore
$((t^A[\bar a] \oplus 0)_i, \frac{1}{2}_i) \not = (1_i, \frac{1}{2}_i)$
in $S_{A_i}$. Since $\beta_i f$ is a $\sqrt{q{\cal PMV}}$-homomorphism, we
have that $t^{S_{A_i}}[\beta_i f (\bar a)] = ((t^A[\bar a] \oplus 0)_i,
\frac{1}{2}_i) \not = (1_i, \frac{1}{2}_i)$ and this is a
contradiction in the case ${\cal S} = {\cal S}^{\Box}$. If ${\cal S}
= {\cal S}^{\circ} $ then $A \in {\cal IP}$. By
Proposition \ref{ALGCIRC}, $\beta_i f(A)$ is a sub algebra of
$D_{A_i}$. Therefore $t^{D_{A_i}}[\beta_i f (\bar a)]=
t^{S_{A_i}}[\beta_i f (\bar a)] \not = (1_i, \frac{1}{2}_i)$ and this is
also a contradiction. Hence ${\cal A}_{{\cal S}} \models t = 1$.

\qed
\end{proof}

\section{Hilbert-style calculus for ${\cal IP}$}
In this section we build a Hilbert-style calculus founded on the
irreversible Poincar\'{e} structure taking into account $PMV$-models
whose logical consequence is based on the preservation of the probability value equal to $1$.

\subsection{Syntaxis and semantic}
Consider the absolutely free algebra $Term_{\cal IP}$ built from the
set of variables $V = \{x_1, x_2 ...\}$ as underling language for
the calculus. In addition we introduce by definition the connective
$\Longleftrightarrow$ as follows: $$\alpha \Longleftrightarrow \beta
\hspace{0.2cm} for \hspace{0.1cm} (\alpha \rightarrow \beta)\odot
(\beta \rightarrow \alpha)$$

Let $A$ be an algebra in ${\cal IP}$ and $p:A \rightarrow
A/_{\equiv}$ be the natural $\langle \oplus, \bullet, \neg, 0,
\frac{1}{2}, 1  \rangle$-homomorphism. Then {\it interpretations} of
the language $Term_{\cal IP}$ in $A$ is any ${\cal IP}$-homomorphism
$e:Term_{\cal IP} \rightarrow A$, and the {\it valuation}
associated to $e$ is the composition $e_p = pe$. Therefore the
$PMV$-models in ${\cal IP}$ are established  and for each $\alpha \in
Term_{\cal IP}$, $e_p(\alpha) = p(e(\alpha))$ represent in this
framework the ``probability value" of the term $\alpha$.

\begin{prop}\label{EQPRB}
Let $D_A$ be the ${\cal S}^{\circ}$-algebra associated to the
$PMV_{\frac{1}{2}}$-chain $A$. If $e$, $e'$ are two interpretations over $D_A$ such that for
each atomic term $\alpha$, $e_p(\alpha) = e'_p(\alpha)$ and
$e_p(\sqrt \alpha) = e'_p(\sqrt \alpha)$ then we have that $e= e'$.
\end{prop}

\begin{proof}
Let $\alpha$ be an atomic term. Suppose that $e(\alpha) = (x,y)$
and  $e'(\alpha) = (x',y')$. Using Proposition \ref{PROJECTION} we
can identify $p$ with the $x-projection$. Thus $e_p(\alpha) = p(x,y)
= x$ and $e'_p(\alpha) = p(x',y') = x'$,  $e_p(\sqrt \alpha) =
p(y,\neg x) = y$ and $e'_p(\sqrt \alpha) = p(y',\neg x') = y'$.
Using the hypothesis we have that $x = x'$ and $y = y'$. Finally by
an inductive argument on the complexity of terms, it results that
$e= e'$.

\qed
\end{proof}

\begin{definition}
{\rm An ${\cal IP}$-term $\alpha$ is a {\it tautology} iff for each interpretation $e$ we have that
$e_p(\alpha) = 1$ }
\end{definition}

Let $Term(\frac{1}{2})$ be the sub-language (without variables) of
$Term_{\cal IP}$ generated by $\langle \oplus, \bullet, \neg, 0,
\frac{1}{2}, 1  \rangle$. For any interpretation $e: Term_{\cal IP}
\rightarrow A$, it is clear that the restriction $e_p:
Term(\frac{1}{2}) \rightarrow A/_\equiv$ is a ${\cal
PMV}_{\frac{1}{2}}$-homomorphism whose image is the sub algebra
$G_{A/_\equiv}(\frac{1}{2})$. By induction on the complexity of
terms we can prove that if $e, e'$ are two interpretations then, for
each $s \in Term(\frac{1}{2})$, $e_p(s) = e'_p(s)$. Consequently by
Remark \ref{CONVENC} each element $s\in Term(\frac{1}{2})$ can be
identified with a single number  $\overline s \in
G_{[0,1]}(\frac{1}{2})$. Taking into account this fact, we introduce
the following axiomatic system:

\begin{definition}
{\rm The following terms are axioms of the ${\cal IP}$-calculus: \\

{\it \L ukasiewicz axioms}

\begin{enumerate}
\item[W1]
$\alpha \rightarrow (\beta \rightarrow \alpha)$

\item[W2]
$(\alpha \rightarrow \beta) \rightarrow ((\beta \rightarrow \gamma ) \rightarrow (\alpha \rightarrow \gamma)) $

\item[W3]
$(\neg \alpha \rightarrow  \neg \beta) \rightarrow (\beta  \rightarrow \alpha) $

\item[W4]
$((\alpha \rightarrow \beta) \rightarrow \beta) \rightarrow ((\beta \rightarrow \alpha) \rightarrow \alpha)$

\end{enumerate}

{\it Constant axioms}
\begin{enumerate}

\item[C1]
$1$

\item[C2]
$\neg 0 \Longleftrightarrow 1$

\item[C3]
$\neg {\bf \frac{1}{2}} \Longleftrightarrow {\bf \frac{1}{2}}$

\item[C4]
$\frac{a}{2^n} \odot \frac{b}{2^m} \Longleftrightarrow \frac{\max \{0, \hspace{0.1cm} a+b2^{n-m}-2^n\}}{{\bf 2}^n}$ with $n \geq m$,

\item[C5]
$\frac{a}{{\bf 2}^n} \bullet \frac{b}{{\bf 2}^m} \Longleftrightarrow \frac{ab}{{\bf 2}^{n+m}} $,

\item[C6]
$\neg \frac{a}{{\bf 2}^n} \Longleftrightarrow \frac{2^n - a}{{\bf 2}^n}$.

\end{enumerate}

{\it Product axioms}

\begin{enumerate}

\item[P1]
$ (\alpha \bullet \beta ) \rightarrow (\beta \bullet \alpha )$

\item[P2]
$ (1 \bullet \alpha ) \Longleftrightarrow \alpha$

\item[P3]
$ (\alpha \bullet \beta ) \rightarrow \beta$

\item[P4]
$(\alpha \bullet \beta ) \bullet \gamma  \Longleftrightarrow \alpha \bullet (\beta \bullet \gamma)$

\item[P5]
$x \bullet (y \odot \neg z) \Longleftrightarrow  (x \bullet y) \odot \neg (x \bullet z) $

\end{enumerate}

{\it Sqrt axioms}

\begin{enumerate}

\item[sQ1]
$\sqrt{\sqrt{\alpha}} \Longleftrightarrow \neg \alpha $

\item[sQ2]
$\sqrt{ \neg \alpha}\Longleftrightarrow \neg \sqrt{\alpha} $

\item[sQ3]
If $*$ is a binary operation $\sqrt{\alpha * \beta} \Longleftrightarrow \frac{1}{2} $

\item[sQ4]
$\sqrt{0} \Longleftrightarrow \sqrt{\frac{1}{2}} \Longleftrightarrow \sqrt{1} \Longleftrightarrow \frac{1}{2}$.

\item[sQ5]
$\{((\frac {1}{4} \bullet \alpha) \oplus ((\frac {1}{4} \bullet
\sqrt \alpha )) \rightarrow s  : \hspace{0.2cm} s\in
Term(\frac{1}{2}), \hspace{0.2cm} \overline s \geq \frac{1+\sqrt2}{4\sqrt2} \}$.

\end{enumerate}
The unique deduction rule is {\it modus ponens} $\{\alpha, \alpha \rightarrow \beta \} \vdash \beta $  (MP).
}
\end{definition}

A theory is any set $T \subseteq Term_{\cal IP}$. A {\it proof} from
$T$ is a sequence of terms $\alpha_1,...,\alpha_n$ such that each
member is either an axiom or a member of $T$ or follows from
preceding members of the sequence by modus ponens. $T \vdash \alpha$
means that $\alpha$ is provable in $T$, that is, $\alpha$ is the
last term of a proof from $T$. Thus the ${\cal IP}$-calculus is
conformed by the pair $\langle Term_{\cal IP}, \vdash \rangle $. If
$T = \emptyset$ we use the notation $\vdash \alpha$ and we said that
$\alpha$ is a {\it theorem}. $T$ is {\it inconsistent} if and only
$T\vdash \alpha $ for each  $\alpha \in Term_{\cal IP}$; otherwise
it is {\it consistent}. We note that axioms W1...W4, C1, C2 and MP
conform the same propositional system as the infinite valued \L
ukasiewicz calculus {\rm  \cite[$\S 4$]{CDM}}.

\begin{lem}\label{BL2}
Let $\alpha, \beta \in Term_{\cal IP}$ and $T$ be a theory. Then the following
items may be proved using only W1...W4, C1, C2, P1...P5 and MP.

\begin{enumerate}

\item
$\vdash \alpha \rightarrow \alpha$

\item
$T \vdash \alpha \odot \beta$ \hspace{0.2 cm} iff \hspace{0.2 cm},  $T \vdash \alpha$ and $T \vdash \beta$,

\item
$T\vdash \alpha \Longleftrightarrow \beta$ \hspace{0.2 cm} iff
\hspace{0.2 cm} $T\vdash \alpha \rightarrow \beta$ and $T\vdash
\beta \rightarrow \alpha$,

\item
$T\vdash \alpha \rightarrow \beta$ and $T\vdash \beta \rightarrow
\gamma$ \hspace{0.2 cm} then \hspace{0.2 cm} $T\vdash \alpha
\rightarrow \gamma$,

\item
$\vdash \neg \neg \alpha \rightarrow \alpha$

\item
$\vdash (\alpha \rightarrow \beta) \rightarrow (\neg \beta \rightarrow \neg \alpha)$,

\item
$\vdash (\alpha \rightarrow \beta) \rightarrow ((\alpha \oplus \gamma )\rightarrow (\beta \oplus \gamma))$,

\item
$\vdash ((\alpha \Longleftrightarrow \beta)\odot (\beta \Longleftrightarrow \gamma)) \rightarrow (\alpha \Longleftrightarrow \gamma)$

\item
$\vdash (\alpha \Longleftrightarrow \beta)\rightarrow ((\alpha \rightarrow \gamma) \Longleftrightarrow (\beta \rightarrow \gamma))$

\item
$\vdash (\alpha \Longleftrightarrow \beta)\rightarrow ((\gamma \rightarrow \alpha ) \Longleftrightarrow (\gamma \rightarrow \beta))$

\item
$\vdash (\alpha \rightarrow \beta) \rightarrow ((\gamma \bullet \alpha )\rightarrow (\gamma \bullet \beta ))$

\end{enumerate}
\end{lem}

\begin{proof}
Items 1...10 are follows from the fact that they are theorems (or meta theorem) in the infinite valued \L ukasiewicz
calculus given in {\rm \cite{HAJ}}. We prove item 11:

\begin{enumerate}

\item[(1)]
$\vdash \gamma \bullet (\alpha \odot \neg \beta) \rightarrow
((\alpha \odot \neg \beta))$ \hspace{3.2 cm} {\footnotesize by Ax
P3}

\item[(2)]
$\vdash ((\gamma\bullet \alpha) \odot \neg(\gamma \bullet
\beta))\rightarrow \gamma \bullet (\alpha \odot \neg \beta)$
\hspace{2 cm} {\footnotesize by Ax P5}

\item[(3)]
$\vdash ((\gamma\bullet \alpha) \odot \neg(\gamma \bullet \beta))
\rightarrow (\alpha \odot \neg \beta)$ \hspace{2.5 cm}{\footnotesize
by 1,2, Ax W2}

\item[(4)]
$\vdash (((\gamma\bullet \alpha) \odot \neg(\gamma \bullet
\beta))\rightarrow (\alpha \odot \neg \beta))   \rightarrow (\neg
(\alpha \odot \neg \beta) \rightarrow \neg ((\gamma\bullet \alpha)
\odot \neg(\gamma \bullet \beta)))$ \noindent

\hspace{8.3 cm} {\footnotesize by Ax W3 }

\item[(5)]
$\vdash \neg (\alpha \odot \neg \beta) \rightarrow \neg
((\gamma\bullet \alpha) \odot \neg(\gamma \bullet \beta))$
\hspace{1.5 cm} {\footnotesize by MP 3,4}

\item[(6)]
$\vdash (\alpha \rightarrow \beta)\rightarrow  \neg(\alpha \odot
\neg \beta)$ \hspace{3 cm} {\footnotesize by def $\odot$, item 1}

\item[(7)]
$\vdash (\alpha \rightarrow \beta)\rightarrow \neg ((\gamma\bullet
\alpha) \odot \neg(\gamma \bullet \beta))$ \hspace{2.1
cm}{\footnotesize by 5,6 , Ax W2}

\item[(8)]
$\vdash \neg ((\gamma\bullet \alpha) \odot \neg(\gamma \bullet
\beta)) \rightarrow ((\gamma\bullet \alpha) \rightarrow (\gamma
\bullet \beta))$ \hspace{0.2 cm} {\footnotesize by def $\odot$, item 1}

\item[(9)]
$\vdash (\alpha \rightarrow \beta)\rightarrow ((\gamma\bullet
\alpha) \rightarrow (\gamma \bullet \beta))$ \hspace{2.4 cm}
{\footnotesize by 7,8, Ax W2}

\end{enumerate}

\qed
\end{proof}

An interpretation  $e$ is a {\it model of a theory} $T$ if and only
if $e_p(\alpha)= 1$ for each $\alpha \in T$. In this case we will
use the notation $e_p(T) = 1$. We use $T \models \alpha$ in case that $e_p(\alpha) = 1$ whenever
$e_p(T) = 1$.

\begin{prop}\label{SOUDNES1}
Axioms of the ${\cal IP}$-calculus are tautologies. Moreover if $e$
is a model for the theory $T$ and $T\vdash \alpha$
then, $e_p(\alpha) = 1$.
\end{prop}

\begin{proof}
The first part is trivial. The second assertion is easily verified from  the fact that the modus ponens preserves valuations equal to $1$.
\qed
\end{proof}

\subsection{The $PMV(\frac{1}{2})$-fragment}
In $Term_{{\cal IP}}$ consider the absolutely free algebra
$Term^{\sqrt{V}}_{PMV(\frac{1}{2})}$ generated by $$\langle V\cup
(\sqrt x)_{x \in V},  \oplus, \bullet, \neg, 0, \frac{1}{2}, 1
\rangle$$ (i.e. taking the family of terms $(\sqrt{x})_{x \in V}$
as atomic terms) together with the calculus given by the axioms
W1 \ldots W4, C1 \ldots C6, P1 \ldots P5 and MP as inference rule.
Proof in $Term^{\sqrt{V}}_{PMV(\frac{1}{2})}$ are denoted by the
symbol $\vdash_{PMV(\frac{1}{2})}$. Note that, for all purposes, the
$PMV(\frac{1}{2})$-fragment given by  $\langle
Term^{\sqrt{V}}_{PMV(\frac{1}{2})},
\vdash_{PMV(\frac{1}{2})}\rangle$ is a $PMV$-calculus. Hence, the
results of Lemma \ref{BL2} continue to be valid in the fragment. Let
$A$ be a $PMV_{\frac{1}{2}}$-algebra. Valuations of
$Term^{\sqrt{V}}_{PMV(\frac{1}{2})}$ in $A$ are
${\cal PMV}_{\frac{1}{2}}$-homorphism $v:
Term^{\sqrt{V}}_{PMV(\frac{1}{2})} \rightarrow A$ where $\sqrt{x}$
is tacked as a variable for each $x\in V$. A term $\alpha \in
Term^{\sqrt{V}}_{PMV(\frac{1}{2})}$ is called
$PMV(\frac{1}{2})$-{\it tautology} if and only if for each valuation
$v$, $v(\alpha) = 1$. Let $T$ be a theory in
$Term^{\sqrt{V}}_{PMV(\frac{1}{2})}$. Then $T$ is said to be {\it
complete} iff, for each pair of terms $\alpha, \beta$ in
$Term^{\sqrt{V}}_{PMV(\frac{1}{2})}$, we have:
$T\vdash_{PMV(\frac{1}{2})} \alpha \rightarrow \beta$ or
$T\vdash_{PMV(\frac{1}{2})} \beta \rightarrow \alpha$.

\begin{lem} \label{EXTCOMP}
Let $T$ be a theory and $\alpha$ be a term, both in
$Term^{\sqrt{V}}_{PMV(\frac{1}{2})}$. Suppose that $T$ does not
prove $\alpha$ in the $PMV(\frac{1}{2})$-fragment. Then there exists
a consistent complete theory $T' \subseteq Term^{\sqrt{V}}_{PMV(\frac{1}{2})}$ such that,
$T'$ is complete, $T \subseteq T'$ and $T'$ does not prove $\alpha$
in the $PMV(\frac{1}{2})$-fragment.
\end{lem}

\begin{proof}
See  {\rm \cite[Lemma 2.4.2]{HAJ})}.
\qed
\end{proof}

\begin{theo} \label{LINDE}
Let $T$ be a consistent theory in the $PMV(\frac{1}{2})$-fragment.
For each term $\alpha \in Term^{\sqrt{V}}_{PMV(\frac{1}{2})}$ we
consider the class $$[\alpha] = \{\beta \in
Term^{\sqrt{V}}_{PMV(\frac{1}{2})} : T\vdash_{PMV(\frac{1}{2})}
\alpha \Longleftrightarrow \beta\}$$ Let $L_T = \{[\alpha]: \alpha
\in Term^{\sqrt{V}}_{PMV(\frac{1}{2})}\}$. If we define the
following operation in $L_T$:

\begin{enumerate}

\item[]
$0 = [0]$  \hspace{2cm} $\neg[\alpha] = [\neg \alpha]$

\item[]
$\frac{1}{2}= [\frac{1}{2}] $ \hspace{2cm} $[\alpha]*[\beta]= [\alpha*\beta]$ for $* \in \{\oplus, \bullet\}$

\item[]
$1 = [1]$

\end{enumerate}

Then  $\langle L_T, \hspace{0.2 cm} \oplus, \bullet, \neg, 0, \frac{1}{2},  1 \rangle$ is a $PMV_{\frac{1}{2}}$-algebra. Moreover if $T$ is a
complete theory then $L_T$ is a totally order set.
\end{theo}

\begin{proof}
We first must see that the operations are well defined on $L_T$. In
the cases $\oplus, \neg, 0, \frac{1}{2}, 1 $ we refer to {\rm
\cite[Lemma 2.3.12]{HAJ}}. The case $\bullet$ follows from Lemma
\ref{BL2}. By axioms W1 \ldots W4, C1 \ldots C5, P1 \ldots P5, it is not very hard to see that $L_T$ is a
$PMV$-algebra. If $T$ is a complete theory, using the same argument as
{\rm \cite[Lemma 2.4.2]{HAJ}}, $L_T$ is a totaly ordered set.

\qed \\
\end{proof}

We will refer to $L_T$ as the {\it Lindenbaum algebra} associated to the theory $T \subseteq Term^{\sqrt{V}}_{PMV(\frac{1}{2})}$.

\subsection{Completeness of the ${\cal IP}$-calculus}

\begin{definition}
{\rm We define the  $PMV(\frac{1}{2})$-{\it translation} $\alpha
\stackrel {t}{\rightarrow} \alpha_t$ as the application $t:
Term_{\cal IP} \rightarrow Term^{\sqrt{V}}_{PMV(\frac{1}{2})}$  such
that:

\begin{enumerate}
\item[]
$x \stackrel {t}{\mapsto} x $ and $ \sqrt x \stackrel {t}{\mapsto} \sqrt x$ \hspace{0.2 cm} for each $x \in V$,

\item[]
$c \stackrel {t}{\mapsto} c $ and $ \sqrt c \stackrel {t}{\mapsto} \frac{1}{2}$ \hspace{0.2 cm} for each $c \in \{0,\frac{1}{2}, 1\}$,

\item[]
$\neg \alpha \stackrel {t}{\mapsto} \neg(\alpha_t)$,

\item[]
$\sqrt {\neg \alpha} \stackrel {t}{\mapsto} (\neg \sqrt {\alpha})_t$,

\item[]
$\sqrt {\sqrt \alpha} \stackrel {t}{\mapsto} (\neg \alpha)_t$,

\item[]
$\sqrt {\alpha \star \beta} \stackrel {t}{\mapsto} \frac{1}{2}$ \hspace{0.2 cm} for each binary connective $\star$,

\item[]
$\alpha \star \beta \stackrel {t}{\mapsto} \alpha_t \star \beta_t $ \hspace{0.2 cm} for each binary connective $\star$,

\end{enumerate}
}
\end{definition}

The $PMV(\frac{1}{2})$-translation is a syntactic representation of
the function $p$ in the $PMV$-model. If $T$ is a theory in
$Term_{\cal IP}$ then, we define the $PMV(\frac{1}{2})$-translation
over $T$ as the set $T_t = \{\alpha_t : \alpha \in T \}$.

\begin{prop} \label{EQUIVTRAS}
Let $\alpha \in Term_{\cal IP}$. Then we have: $$\vdash \alpha \Longleftrightarrow \alpha_t$$

\end{prop}

\begin{proof}
We use induction on complexity of terms. Let $\alpha$ be an atomic term. By definition of $PMV(\frac{1}{2})$-translation, Lemma \ref{BL2}-1 and axiom sQ4 of the ${\cal IP}$-calculus it is clear that  $\vdash \alpha \Longleftrightarrow \alpha_t$ and $\vdash \sqrt{\alpha} \Longleftrightarrow (\sqrt{\alpha})_t$. Suppose that $\vdash \alpha \Longleftrightarrow \alpha_t$ and $\vdash \beta \Longleftrightarrow \beta_t$.

\begin{itemize}
\item
By Lemma \ref{BL2}-6 we have that $\vdash \neg \alpha \Longleftrightarrow \neg \alpha_t$.

\item
Let $\star \in \{ \oplus, \bullet \}$. Then we have that:

\begin{enumerate}
\item[(1)]
$\vdash \alpha \rightarrow \alpha_t$

\item[(2)]
$\vdash (\alpha \rightarrow \alpha_t) \rightarrow ((\alpha \star
\beta) \rightarrow (\alpha_t \star \beta))$ \hspace{0.2 cm}
{\footnotesize by Lemma \ref{BL2}, item 7 or 11}

\item[(3)]
$\vdash (\alpha \star \beta) \rightarrow (\alpha_t \star \beta)$ \hspace{2.5 cm} {\footnotesize MP 1-2}

\item[(4)]
$\vdash \beta \rightarrow \beta_t$

\item[(5)]
$\vdash (\beta \rightarrow \beta_t) \rightarrow ((\alpha_t \star
\beta) \rightarrow (\alpha_t \star \beta_t))$  \hspace{0.2 cm}
{\footnotesize by Lemma \ref{BL2}, item 7 or 11}

\item[(6)]
$\vdash (\alpha_t \star \beta) \rightarrow (\alpha_t \star \beta_t)$ \hspace{2.5 cm} {\footnotesize MP 4-5}

\item[(7)]
$\vdash (\alpha \star \beta) \rightarrow (\alpha_t \star \beta_t)$  \hspace{0.2 cm} {\footnotesize by Lemma \ref{BL2}-4}

\end{enumerate}

By the same argument we can prove that $\vdash (\alpha_t \star
\beta_t) \rightarrow (\alpha \star \beta) $. Hence $\vdash (\alpha
\star \beta) \Longleftrightarrow (\alpha \star \beta)_t$.

\item
If $\alpha$ is $\sqrt{\gamma}$ then we must consider three cases:

i) $\gamma$ is $\gamma_1 \star \gamma_2$ such that $\star \in \{
\oplus, \bullet \}$. Then $\alpha_t = (\sqrt{\gamma})_t =
(\sqrt{\gamma_1 \star \gamma_2})_t = \frac{1}{2}$. By Axiom sQ3,
$(\sqrt{\gamma}) \Longleftrightarrow \frac{1}{2}$. Hence $\vdash
\alpha \Longleftrightarrow \alpha_t$.

ii) $\gamma$ is $\neg \gamma_1$. Then $\alpha_t = (\sqrt{\gamma})_t
= (\sqrt{\neg \gamma_1})_t = (\neg \sqrt{\gamma_1})_t =
\neg(\sqrt{\gamma_1})_t$. By inductive hypothesis $\vdash
\sqrt{\gamma_1} \Longleftrightarrow (\sqrt{\gamma_1})_t$ and then
$\vdash \neg \sqrt{\gamma_1} \Longleftrightarrow \neg
(\sqrt{\gamma_1})_t$. By Axiom sQ2, $\vdash \sqrt{\neg \gamma_1}
\Longleftrightarrow \neg \sqrt{\gamma_1} $. Thus $\vdash \alpha
\Longleftrightarrow \alpha_t$.

iii)  $\gamma$ is $\sqrt {\gamma_1}$. Then $\alpha_t = (\sqrt{\sqrt
{\gamma_1}})_t = (\neg \gamma_1 )_t = \neg(\gamma_1)_t$. By
inductive hypothesis $\vdash \gamma_1 \Longleftrightarrow
(\gamma_1)_t$ and then $\vdash \neg \gamma_1 \Longleftrightarrow
\neg(\gamma_1)_t$. By Axiom sQ1, $\vdash \sqrt{\sqrt {\gamma_1}}
\Longleftrightarrow \neg \gamma_1$. Thus $\vdash \alpha
\Longleftrightarrow \alpha_t$.

\end{itemize}
\qed
\end{proof}

Taking into account the axiom sQ5, we define the following theory
which plays an important role in relation to deductions on the
${\cal IP}$-calculus with respect to deductions in the
$PMV(\frac{1}{2})$-fragment.

\begin{definition}
{\rm We consider the following three groups of terms in $Term^{\sqrt{V}}_{PMV(\frac{1}{2})}$\\

$T_1 = \{((\frac {1}{4} \bullet x) \oplus ((\frac {1}{4} \bullet \sqrt x )) \rightarrow s
: x\in V \cup \{0,\frac{1}{2} ,1\} , \hspace{0.2cm} \overline s \geq \frac{1 + \sqrt 2}{4\sqrt 2} \}$,\\

$T_2 = \{((\frac {1}{4} \bullet \neg x) \oplus ((\frac {1}{4} \bullet \neg \sqrt x )) \rightarrow
s : x\in V\cup \{0,\frac{1}{2} ,1\}, \hspace{0.2cm} \overline s \geq \frac{1 + \sqrt 2}{4\sqrt 2} \}$,\\

$T_3 = \{((\frac {1}{4} \bullet \neg x) \oplus ((\frac {1}{4} \bullet \sqrt x )) \rightarrow s :
x\in V\cup \{0,\frac{1}{2} ,1\}, \hspace{0.2cm} \overline s \geq \frac{1 + \sqrt 2}{4\sqrt 2} \}$,\\

$T_4 = \{((\frac {1}{4} \bullet x) \oplus ((\frac {1}{4} \bullet \neg \sqrt x )) \rightarrow s : x\in V\cup
\{0,\frac{1}{2} ,1\}, \hspace{0.2cm} \overline s \geq \frac{1 + \sqrt 2}{4\sqrt 2} \}$.\\

\noindent
Then we define: $$T_D = T_1 \cup T_2 \cup T_3 \cup T_4$$
}
\end{definition}

\begin{prop} \label{PRETRAS}
Let $\alpha \in Term_{\cal IP}$ and $s\in Term(\frac{1}{2})$. If $\overline s \geq \frac{1 + \sqrt 2}{4\sqrt 2}$ then we have:

\begin{enumerate}
\item
$T_D \vdash_{PMV(\frac{1}{2})} (((\frac {1}{4} \bullet \alpha) \oplus (\frac {1}{4} \bullet \sqrt \alpha)) \rightarrow
 s)\hspace{0.1cm} _t $ \hspace{0.6cm} noted \hspace{0.1cm} $T_D\vdash_{PMV(\frac{1}{2})} \alpha^1_t $

\item
$T_D\vdash_{PMV(\frac{1}{2})} (((\frac {1}{4} \bullet \neg \alpha) \oplus (\frac {1}{4} \bullet \neg \sqrt \alpha))
 \rightarrow s)\hspace{0.1cm} _t $ \hspace{0.1cm} noted \hspace{0.1cm} $T_D\vdash_{PMV(\frac{1}{2})} \alpha^2_t $

\item
$T_D\vdash_{PMV(\frac{1}{2})} (((\frac {1}{4} \bullet \neg \alpha) \oplus (\frac {1}{4} \bullet \sqrt \alpha)) \rightarrow
 s)\hspace{0.1cm} _t$ \hspace{0.4cm} noted \hspace{0.1cm} $T_D\vdash_{PMV(\frac{1}{2})} \alpha^3_t $

\item
$T_D\vdash_{PMV(\frac{1}{2})} (((\frac {1}{4} \bullet \alpha) \oplus (\frac {1}{4} \bullet \neg \sqrt \alpha))
\rightarrow  s)\hspace{0.1cm} _t$ \hspace{0.4cm} noted \hspace{0.1cm} $T_D\vdash_{PMV(\frac{1}{2})} \alpha^4_t $\\

\end{enumerate}

\end{prop}

\begin{proof}
We use induction on complexity of $\alpha$. The case $\alpha \in V
\cup \{0,\frac{1}{2} ,1\}$ is immediate from $T_1$. In particular if
$\alpha$ is $1$ then we have that $(((\frac {1}{4} \bullet 1)
\oplus (\frac {1}{4} \bullet \sqrt 1)) \rightarrow
s)\hspace{0.1cm} _t = ((\frac {1}{4} \bullet 1) \oplus(\frac {1}{4}
\bullet \frac {1}{2})) \rightarrow  s$. Therefore it is not very
hard to see that: $$T_D \vdash_{PMV(\frac{1}{2})} (\frac {1}{4} \oplus \frac
{1}{8} ) \rightarrow  s$$

\noindent
Suppose $\alpha$ is $\alpha_1 \star \alpha_2$ such that $\star \in \{\oplus, \bullet \}$. \\

By Axiom C3 and Lemma \ref{BL2} it follows that $\vdash_{PMV(\frac{1}{2})}
\alpha^1_t \Longleftrightarrow \alpha^4_t $ and $\vdash_{PMV(\frac{1}{2})} \alpha^2_t \Longleftrightarrow
\alpha^3_t $. Taking into account that $\alpha^1_t$ is $((\frac
{1}{4} \bullet \alpha_t) \oplus \frac {1}{8}) \rightarrow  s $ and
$\alpha^3_t$ is $((\frac {1}{4} \bullet \neg \alpha_t) \oplus
\frac{1}{8}) \rightarrow  s $, we consider the term $((\frac
{1}{4} \bullet \beta) \oplus \frac {1}{8}) \rightarrow s $ where
$\beta \in \{\alpha_t, \neg \alpha_t \}$. Therefore we need to see
that $T_D \vdash_{PMV(\frac{1}{2})} ((\frac {1}{4} \bullet \beta) \oplus
\frac{1}{8}) \rightarrow s $. In fact:

\begin{enumerate}
\item[(1)]
$T_D \vdash_{PMV(\frac{1}{2})} (\frac {1}{4} \oplus \frac {1}{8} ) \rightarrow  s $

\item[(2)]
$\vdash_{PMV(\frac{1}{2})} (\frac {1}{4} \bullet \beta) \rightarrow \frac{1}{4}$   \hspace{0.1 cm} {\footnotesize by Ax P3}

\item[(3)]
$\vdash_{PMV(\frac{1}{2})} ((\frac {1}{4} \bullet \beta) \rightarrow \frac{1}{4})
\rightarrow (((\frac {1}{4} \bullet \beta) \oplus
\frac{1}{8})\rightarrow (\frac{1}{4} \oplus \frac{1}{8})) $
\hspace{0.2 cm} {\footnotesize by Lemma \ref{BL2}-7 }

\item[(4)]
$\vdash_{PMV(\frac{1}{2})} ((\frac {1}{4} \bullet \beta) \oplus
\frac{1}{8})\rightarrow (\frac{1}{4} \oplus \frac{1}{8})$
\hspace{0.2 cm} {\footnotesize by MP 2,3}

\item[(5)]
$\vdash_{PMV(\frac{1}{2})} ((\frac{1}{4} \oplus \frac{1}{8}) \rightarrow s)
\rightarrow (((\frac {1}{4} \bullet \beta) \oplus \frac{1}{8}))
\rightarrow s) $ \hspace{0.2 cm} {\footnotesize by Ax W2, 4 and MP}

\item[(6)]
$T_D \vdash_{PMV(\frac{1}{2})} ((\frac {1}{4} \bullet \beta) \oplus \frac{1}{8}) \rightarrow s $ \hspace{2.5 cm} {\footnotesize MP 1,5}.

\end{enumerate}

\noindent
Suppose $\alpha$ is $\neg{\beta}$.

\begin{enumerate}
\item
$\alpha^1_t = (((\frac {1}{4} \bullet \neg \beta) \oplus (\frac
{1}{4} \bullet \sqrt {\neg  \beta})) \rightarrow \overline
s)\hspace{0.1cm} _t = ((\frac {1}{4} \bullet \neg  \beta_t) \oplus
(\frac {1}{4} \bullet \neg (\sqrt \beta)_t)) \rightarrow \overline s
= $

$(((\frac {1}{4} \bullet \neg  \beta) \oplus (\frac {1}{4} \bullet
\neg (\sqrt \beta))) \rightarrow \overline s)\hspace{0.1cm} _t =
\beta^2_t$. Then, by inductive hypothesis we have that, $T_D
\vdash_{PMV(\frac{1}{2})}\beta^2_t$.

\item
$\alpha^2_t = (((\frac {1}{4} \bullet \neg \neg \beta) \oplus (\frac
{1}{4} \bullet \neg \sqrt {\neg  \beta})) \rightarrow \overline
s)\hspace{0.1cm} _t = $

$((\frac {1}{4} \bullet \neg \neg \beta_t) \oplus (\frac {1}{4}
\bullet \neg \neg (\sqrt \beta)_t)) \rightarrow \overline s $. By
Proposition \ref{BL2} we have:

$\vdash_{PMV(\frac{1}{2})} ((\frac {1}{4} \bullet \neg \neg \beta_t) \oplus
(\frac {1}{4} \bullet \neg \neg (\sqrt \beta)_t)) \rightarrow
\overline s \Longleftrightarrow ((\frac {1}{4} \bullet \beta_t) \oplus (\frac
{1}{4} \bullet (\sqrt \beta)_t)) \rightarrow \overline s $

and $((\frac {1}{4} \bullet \beta_t) \oplus (\frac {1}{4} \bullet
(\sqrt \beta)_t)) \rightarrow \overline s = \beta^1_t$. Then, by
inductive hypothesis we have that, $T_D \vdash_{PMV(\frac{1}{2})}\beta^1_t$.

\item
$\alpha^3_t = (((\frac {1}{4} \bullet \neg \neg \beta) \oplus (\frac
{1}{4} \bullet \sqrt {\neg  \beta})) \rightarrow \overline
s)\hspace{0.1cm} _t =  ((\frac {1}{4} \bullet \neg \neg \beta_t)
\oplus (\frac {1}{4} \bullet \neg (\sqrt \beta)_t)) \rightarrow
\overline s $

By Proposition \ref{BL2} we have:

$\vdash_{PMV(\frac{1}{2})} ((\frac {1}{4} \bullet \neg \neg \beta_t) \oplus
(\frac {1}{4} \bullet \neg (\sqrt \beta)_t)) \rightarrow \overline s
\Longleftrightarrow ((\frac {1}{4} \bullet \beta_t) \oplus (\frac {1}{4} \bullet
\neg (\sqrt \beta)_t)) \rightarrow \overline s $

where $((\frac {1}{4} \bullet \beta_t) \oplus (\frac {1}{4} \bullet
\neg (\sqrt \beta)_t)) \rightarrow \overline s = \beta^4_t$. Then,
by inductive hypothesis we have $T_D \vdash_{PMV(\frac{1}{2})}\beta^4_t$.

\item
$\alpha^4_t = (((\frac {1}{4} \bullet  \neg \beta) \oplus (\frac
{1}{4} \bullet \neg \sqrt {\neg  \beta})) \rightarrow \overline
s)\hspace{0.1cm} _t = ((\frac {1}{4} \bullet \neg \beta_t) \oplus
(\frac {1}{4} \bullet \neg \neg (\sqrt \beta)_t)) \rightarrow
\overline s $ By Proposition \ref{BL2} we have:

$\vdash_{PMV(\frac{1}{2})} ((\frac {1}{4} \bullet \neg \beta_t) \oplus (\frac
{1}{4} \bullet \neg \neg (\sqrt  \beta)_t)) \rightarrow \overline s
\Longleftrightarrow ((\frac {1}{4} \bullet \neg \beta_t) \oplus (\frac {1}{4}
\bullet (\sqrt \beta)_t)) \rightarrow \overline s $

where $((\frac {1}{4} \bullet \neg \beta_t) \oplus (\frac {1}{4}
\bullet (\sqrt \beta)_t)) \rightarrow \overline s = \beta^3_t$.
Then, by inductive hypothesis we have $T_D \vdash_{PMV(\frac{1}{2})}\beta^3_t$.

\end{enumerate}

\noindent
Suppose $\alpha$ is $\sqrt \beta$.

\begin{enumerate}

\item
$\alpha^1_t = (((\frac {1}{4} \bullet  \sqrt \beta) \oplus (\frac
{1}{4} \bullet \sqrt {\sqrt  \beta})) \rightarrow \overline
s)\hspace{0.1cm} _t =  ((\frac {1}{4} \bullet \sqrt \beta_t) \oplus
(\frac {1}{4} \bullet \neg \beta_t)) \rightarrow \overline s $

But using Proposition \ref{BL2}

$\vdash_{PMV(\frac{1}{2})} ((\frac {1}{4} \bullet \sqrt \beta_t) \oplus (\frac
{1}{4} \bullet \neg \beta_t)) \rightarrow \overline s    \Longleftrightarrow
((\frac {1}{4} \bullet \neg \beta_t) \oplus (\frac {1}{4} \bullet
\sqrt \beta_t)) \rightarrow \overline s $

where $ ((\frac {1}{4} \bullet \neg \beta_t) \oplus (\frac {1}{4}
\bullet \sqrt \beta_t)) \rightarrow \overline s = \beta^3_t$. By inductive hypothesis we have $T_D \vdash_{PMV(\frac{1}{2})}\beta^3_t$. \\

For the rest of this case, i.e. $\alpha^2_t, \alpha^3_t,
\alpha^4_t$, it follows in a similar way.

\end{enumerate}

\qed
\end{proof}

\begin{theo}\label{DEDTRAS}
Let $T$ be a theory and $\alpha$ be a term both in $Term_{\cal IP}$. Then we have:

$$T \vdash \alpha \hspace{0.4 cm} iff \hspace{0.4 cm} T_t \cup T_D  \vdash_{PMV(\frac{1}{2})} \alpha_t $$

\end{theo}

\begin{proof}
{\rm Suppose that $T \vdash \alpha$. We use induction on the length
of the proof of $\alpha$ noted by $Length(\alpha)$. If
$Length(\alpha) = 1$  then we have the following possibility:

\begin{enumerate}
\item
$\alpha$ is one of axioms W1, $\cdots$ , W4, \hspace {0.2 cm} C1 $\cdots$
C6 \hspace {0.2 cm} P1 , $\cdots$ , P5 . In this case $\alpha_t$ result an axiom of the $PMV(\frac{1}{2})$-fragment.

\item
$\alpha$ is one of the axioms sQ1, $\cdots$, sQ4. In this case
$\alpha_t$ looks like $\beta \Longleftrightarrow \beta$ in the $PMV$-fragment and
by Proposition \ref{BL2} 1 and 9 this terms are $PMV$-theorems .

\item
If $\alpha$ is an axiom sQ5 then we use Proposition \ref{PRETRAS} resulting $T_{Q5}\vdash \alpha_t$

\item
If $\alpha \in T$ it is clear that $\alpha_t \in T_t$.

\end{enumerate}

Suppose that the theorem is valid for $Length(\alpha) < n$. We consider  $Lengh(\alpha) = n$.
Thus we have a proof of $\alpha$ from $T$ as follows  $$\alpha_1, \cdots, \alpha_m \rightarrow
\alpha, \cdots, \alpha_m , \cdots, \alpha_{n-1}, \alpha$$ obtaining $\alpha$ by MP
from $\alpha_m \rightarrow \alpha$ and $\alpha_m$. Using inductive hypothesis we have $T_t \cup T_D
\vdash_{PMV_{\frac{1}{2}}} (\alpha_m \rightarrow \alpha)_t$ and  $T_t \cup T_D  \vdash_{PMV(\frac{1}{2})} (\alpha_m)_t$. Taking
into account that $(\alpha_m \rightarrow \alpha)_t$ is $(\alpha_m)_t \rightarrow \alpha_t$, by MP we have $T_t \cup T_D  \vdash_{PMV(\frac{1}{2})} \alpha_t$. \\

For the converse, suppose that $T_t \cup T_D  \vdash_{PMV(\frac{1}{2})}
\alpha_t$. Then there exist two subsets $\{\beta_1, \cdots, \beta_n
\} \subseteq T$ and  $\{\gamma_1, \cdots, \gamma_m \} \subseteq T_D$
such that $$\{(\beta_1)_t, \cdots, (\beta_n)_t , \gamma_1, \cdots,
\gamma_m \} \vdash_{PMV(\frac{1}{2})} \alpha_t $$ Consequently $\{(\beta_1)_t,
\cdots, (\beta_n)_t , \gamma_1, \cdots, \gamma_m \} \vdash \alpha_t
$. By Lemma \ref{EQUIVTRAS} we have that $\vdash \alpha \equiv
\alpha_t $ and $\vdash \beta_i \equiv (\beta_i)_t $ for each $i \in
\{1, \cdots, n\}$. Moreover, by Axiom sQ5, it is not very hard to see
that $\vdash \gamma_j$ for each for each $j \in \{1, \cdots, m\}$.
Thus $\{\beta_1, \cdots, \beta_n \} \vdash \alpha$ and $T\vdash
\alpha$. }

\qed
\end{proof}

\begin{coro} \label{DEDTRASTEO}
Let $\alpha \in Term_{\cal IP}$. Then we have  $ \vdash \alpha$ iff $T_D  \vdash_{PMV(\frac{1}{2})} \alpha_t $

\qed
\end{coro}

\noindent
Let $S_A$ be a ${\cal S}^\Box$-algebra from the ${\cal PMV}_\frac{1}{2}$-chain $A$. Consider the ${\cal S}^\circ$-algebra given by the sub algebra $D_A$
of $S_A$. We introduce the following sets: $$E_{D_A} = \{\mbox{interpretations} \hspace{0.1cm} e: Term_{\cal IP} \rightarrow  D_A \}$$ $$V_D = \{\mbox{${\cal PMV}_\frac{1}{2}$-homorphisms}\hspace{0.1cm} v: Term^{\sqrt{V}}_{PMV(\frac{1}{2})} \rightarrow Reg(D_A) \hspace{0.2cm} s.t. \hspace{0.1cm} v(T_D)= 1 \}$$

\begin{prop} \label{SEMTRAS}
Let $e \in E_{D_A}$ and the restriction $v_e = e_p \mid_{Term^{\sqrt{V}}_{PMV(\frac{1}{2})}}$. Then the assignment $e \mapsto v_e$
is a bijection $E_{D_A} \rightarrow V_D$ such that $e_p(\alpha) = v_e(\alpha_t)$.
\end{prop}

\begin{proof}
We will see that $e \rightarrow v_e$ is well defined in the sense that $v_e \in V_D$.
Let $\alpha \in T_D$. Since $T_D \subseteq Term^{\sqrt{V}}_{PMV(\frac{1}{2})}$ then
$v_e(\alpha) = e_p(\alpha)$ and  $v_e(\alpha) =
(\frac{a}{4} \oplus \frac{\sqrt a}{4}) \rightarrow \overline s$ for some $a \in
D_A$ and $\overline s \geq \frac{ 1 + \sqrt 2}{4 \sqrt 2}$. Since $D_A \in {\cal IP}$ then,
$(\frac{a}{4} \oplus  \frac{\sqrt a}{4}) \leq \frac{ 1 + \sqrt 2}{4 \sqrt 2}$ resulting $v_e(\alpha) = 1$. Hence $v_e(T_D) = 1$. \\

Suppose that $v_{e_1} = v_{e_2}$. Let $t$ be an atomic term in
$term_{\cal IP}$. Then we have that ${e_1}_p(t) = v_{e_1}(t) =
v_{e_2}(t) = {e_2}_p(t)$ and ${e_1}_p(\sqrt{t}) = v_{e_1}(\sqrt{t})
= v_{e_2}(\sqrt{t}) = {e_2}_p(\sqrt{t})$. Therefore by Proposition
\ref{EQPRB}, $e_1 = e_2$ and $e \mapsto v_e$ is injective.

Now we will prove the surjectivity. Let $v \in V_D$. For each atomic
term $t$ in $Term_{\cal IP}$ we define the interpretation
$e:Term_{\cal IP} \rightarrow D_A$: such that $e(t) = (v(t), v(\sqrt
t))$ for each atomic term $t$. By induction on complexity of terms
we prove that $v_e = v$. For atomic terms in ${\cal IP}$ it follows
by definition of $e$. If $t$ is an atomic terms then $e(\sqrt t) =
(v(\sqrt t), \neg v(t))$ and we have that $v_e(\sqrt t) = v(\sqrt
t)$. That constitutes the base of the induction in the language
$Term^{\sqrt{V}}_{PMV(\frac{1}{2})}$. Now let our claim hold
whenever the complexity of term is less than $n$ and $\alpha$ have
complexity $n$

\begin{itemize}
\item
if $\alpha \in Term^{\sqrt{V}}_{PMV(\frac{1}{2})}$ is $\alpha_1
\star \alpha_2$ where $\star \in \{\oplus, \bullet\}$ then we have
that $e(\alpha) = e(\alpha_1)\star e(\alpha_1) = (v(\alpha_1)\star
v(\alpha_2), \frac{1}{2})$ and $v_e(\alpha) = v(\alpha_1)\star
v(\alpha_2) = v(\alpha_1 \star \alpha_2) = v(\alpha)$.

\item
if $\alpha \in Term^{\sqrt{V}}_{PMV(\frac{1}{2})}$ is $\neg
\alpha_1$ then we have that $e(\alpha) = \neg e(\alpha_1)$ and
$v_e(\alpha) = \neg v_e(\alpha_1) = \neg v(\alpha_1) = v(\alpha)$.
\end{itemize}

\noindent
Thus $v = v_e$ and $e \mapsto v_e$ is a bijection from $E_{D_A}$ onto $V_D$.

Let $e\in E_{D_A}$. By induction on complexity of terms we prove
that for each $\alpha \in Term_{\cal IP}$, $e_p(\alpha) =
v_e(\alpha_t)$. Let $\alpha$ be an atomic term then $e_p(\alpha) =
e_p(\alpha_t) = v_e(\alpha_t)$. Now let our claim hold whenever the
complexity of term is less than $n$ and $\alpha$have complexity $n$.
If $\alpha$ is $\alpha_1 \star \alpha_2$ where $\star \in \{\oplus,
\bullet\}$ or $\alpha$ is $\neg \alpha_1$, this case is routine.
Suppose that $\alpha$ is $\sqrt{\alpha_1}$. Let us consider the
following cases:

\begin{itemize}
\item
If $\alpha_1$ is an atomic term. Then its follows from the fact that $({\sqrt{\alpha_1}})_t = \sqrt{\alpha_1}$.

\item
$\alpha$ is $\sqrt{\neg \alpha_1}$. Then $e_p(\alpha) =
e_p(\sqrt{\neg \alpha_1}) = \neg e_p(\sqrt{\alpha_1}) = \neg
v_e((\sqrt{\alpha_1})_t) = v_e(\neg (\sqrt{\alpha_1})_t) =
v_e((\sqrt{\neg \alpha_1})_t) = v_e(\alpha_t)$.

\item
$\alpha$ is $\sqrt{\sqrt{\alpha_1}}$. Then $e_p(\alpha) =
e_p(\sqrt{\sqrt{\alpha_1}}) = \neg e_p(\alpha_1) = \neg
v_e({{\alpha_1}}_t) = v_e(\neg {{\alpha_1}}_t) =
v_e((\sqrt{\sqrt{\alpha_1}})_t) = v_e(\alpha)$.

\item
$\alpha_1$ is $\sqrt{\alpha_2 \star \alpha_3 }$ where $\star \in
\{\oplus, \bullet\}$. Then $e_p(\alpha) = e_p(\sqrt{\alpha_2 \star
\alpha_3 }) = (\frac{1}{2}, \frac{1}{2}) = e_p{(\frac{1}{2})} =
v_e((\sqrt{\alpha_2 \star \alpha_3})_t) = v_e(\alpha_t)$.

\end{itemize}

\noindent
Hence $e_p(\alpha)= v_e(\alpha_t)$ for each $\alpha \in Term_{\cal IP}$.

\qed
\end{proof}

\begin{theo} \label{COMPLETENESS}
Let $T$ be a theory and $\alpha$ be a term both in $Term_{\cal IP}$ then $$T \models \alpha \hspace{0.4 cm} iff \hspace{0.4 cm} T  \vdash \alpha$$

\end{theo}

\begin{proof}
We assume that $T$ is consistent. Suppose that $T \models \alpha$
but $T$ does not prove $\alpha$. By Theorem \ref{DEDTRAS} $T_t \cup
T_D$ does not prove $\alpha_t$ in the $PMV(\frac{1}{2})$-fragment. In view of Lemma
\ref{EXTCOMP} and Theorem \ref{LINDE} there exists a theory $T'$ in
the $PMV(\frac{1}{2})$-fragment such that $T_t \cup T_D \subseteq T'$, $T'$ does
not prove $\alpha_t$ and $L_{T'}$ is a totally ordered $PMV_{\frac{1}{2}}$-algebra.
Thus $[\alpha_t] \not= 1$. If we consider the natural $PMV_{\frac{1}{2}}$-valuation
$v:  Term^{\sqrt{V}}_{PMV(\frac{1}{2})} \rightarrow L_{T'}$ then $[\alpha_t]  = v(\alpha_t)  \not= 1$. By
Proposition \ref{SEMTRAS}, there exits an interpretation $e: Term_{\cal IP} \rightarrow
D_{L_{T'}}$ such that $e_p(\alpha) = v(\alpha_t) \not = 1$ which is
a contradiction since $e_p(T) = 1$. The converse is immediate.

\qed
\end{proof}

\vspace{0.2cm}

Now we can establish a compactness theorem for the quantum gates
logic

\begin{theo} \label{COMPAC}
Let $T$ be a theory and $\alpha$ be a term both in $Term_{\cal IP}$. Then we have:

$$T\models \alpha \hspace{0.2cm} \mbox{iff} \hspace{0.2cm} \exists \hspace{0.1cm} T_0
\subseteq T \mbox{finite} \hspace{0.1cm} \mbox{such that}
\hspace{0.2cm} T_0\models \alpha $$

\end{theo}

\begin{proof}
If $T\models \alpha$ by Theorem \ref{COMPLETENESS} there exists a proof of
$\alpha$, $\alpha_1, \cdots \alpha_n, \alpha$ from $T$. If we
consider $T_0 = \{ \alpha_k \in T : \alpha_k \in  \{\alpha_1, \cdots
\alpha_n\}\}$ then $T_0\models \alpha $. The converse is immediate.
\qed
\end{proof}

``

{\small \noindent Hector Freytes e-mail: hfreytes@gmail.com \hspace{0.3cm } hfreytes@dm.uba.ar }

{\small \noindent Graciela Domenech e-mail: domenech@iafe.uba.ar }

\end{document}